\documentclass[letterpaper,journal]{IEEEtran}

\usepackage{amsmath,amsfonts}
\usepackage{algorithmic}
\usepackage{algorithm}
\usepackage{array}
\usepackage{multirow}
\usepackage{booktabs,makecell}
\usepackage[caption=false,font=normalsize,labelfont=sf,textfont=sf]{subfig}
\usepackage{textcomp}
\usepackage{stfloats}
\usepackage{url}
\usepackage{verbatim}
\usepackage{graphicx}
\usepackage{cite}
\usepackage[hidelinks]{hyperref}
\usepackage{orcidlink}

\newtheorem{lemma}{Lemma}
\newtheorem{theorem}{Theorem}
\newtheorem{corollary}{Corollary}
\newtheorem{proposition}{Proposition}
\DeclareCaptionFont{axislabel}{\fontsize{2.1pt}{2.7pt}\selectfont}

\DeclareUnicodeCharacter{2032}{\textquotesingle}

\renewcommand{\footnoterule}{%
  \ifnum\value{page}>1
    \kern-3pt
    \hrule width 0.4\columnwidth height 0.4pt
    \kern2.6pt
  \fi
}

\hypersetup{
  pdftitle={On the Relationship between Model Quantization and Model Inversion Attacks},
  pdfauthor={Rongke Liu, Youwen Zhu}
}

\begin{document}

\title{On the Relationship between Model Quantization and Model Inversion Attacks}

\author{Rongke Liu$^{\orcidlink{0000-0003-4204-8793}}$,
Youwen Zhu$^{\orcidlink{0000-0003-4365-9713}}$%
\thanks{Manuscript received xx xx, 2026.
This work was supported by the National Natural Science Foundation of China
(No. xx), the Jiangsu Provincial Key Research and
Development Program (No. xx), and the
Postgraduate Research \& Practice Innovation Program of Jiangsu Province
(No. 26CXJH1036).
\textit{(Corresponding author: Youwen Zhu.)}}%
\thanks{Rongke Liu and Youwen Zhu are with the School of Computer Science
and Technology, Nanjing University of Aeronautics and Astronautics,
Nanjing 211106, China (e-mail: \{liurongke, zhuyw\}@nuaa.edu.cn).}%
\thanks{We thank the Institute of Automation, Chinese Academy of Sciences (CASIA), for the palmprint and iris datasets.}%
\thanks{Code is available at xx.}%
\thanks{Digital Object Identifier xx.}}

\markboth{IEEE Transactions on Pattern Analysis and Machine Intelligence}%
{Liu and Zhu: On the Relationship between Model Quantization and Model Inversion Attacks}

\maketitle

\begin{abstract}
Model quantization reduces the numerical precision of neural network weights and activations to lower storage and computational costs.
Since model inversion attacks use outputs or intermediate features to reconstruct training data or inference inputs, quantization may also alter their effectiveness.
However, this relationship remains insufficiently understood, leaving two questions unresolved:
\textit{How does model quantization affect model inversion?}
\textit{How do data characteristics influence this relationship?}
To address the first question, we bound quantization-induced changes in mutual information between inputs and a categorical variable defined by prediction probabilities, distinguishing informational effects from attack optimization obstacles.
To address the second, we identify data-dependent changes in feature distributions and inversion outcomes, with pronounced quantization sensitivity differences at 4 bits.
Building on these insights, we propose a privacy-aware post-training quantization method to improve inversion resistance while recovering utility.
It estimates layer-wise task sensitivity for candidate bit widths using a Fisher-type proxy to guide budget-aware bit allocation, calibrates activation ranges, and jointly optimizes weight and activation quantization scales and weight-rounding decisions using task-recovery and geometry-retention objectives with scale and rounding regularization.
Experiments use multiple metrics across face, palmprint, and iris recognition tasks and different neural network architectures.
For ResNet-50, Palm--4-bit reduces RL-MIA's strict success from 54\% to 26\%, while accuracy decreases from 99.01\% to 96.55\%, relative to FP32.
Moreover, our defense is orthogonal to output post-processing: adding Stealthy Shield Defense (SSD, $\epsilon = 0.1$) to Iris--4.5-bit further reduces BREP-MI's strict success from 63.33\% to 37.33\%, while accuracy decreases from 92.8\% to 87.6\%, relative to quantization alone.
\end{abstract}


% Keep the wide background figure at the top of page 2 as the body develops.
% Zero top glue also top-aligns this figure if LaTeX uses a float-only page.
\makeatletter
\setlength{\@dblfptop}{0pt}
\makeatother
\begin{figure*}[!t]
\centering
\includegraphics[width=\textwidth]{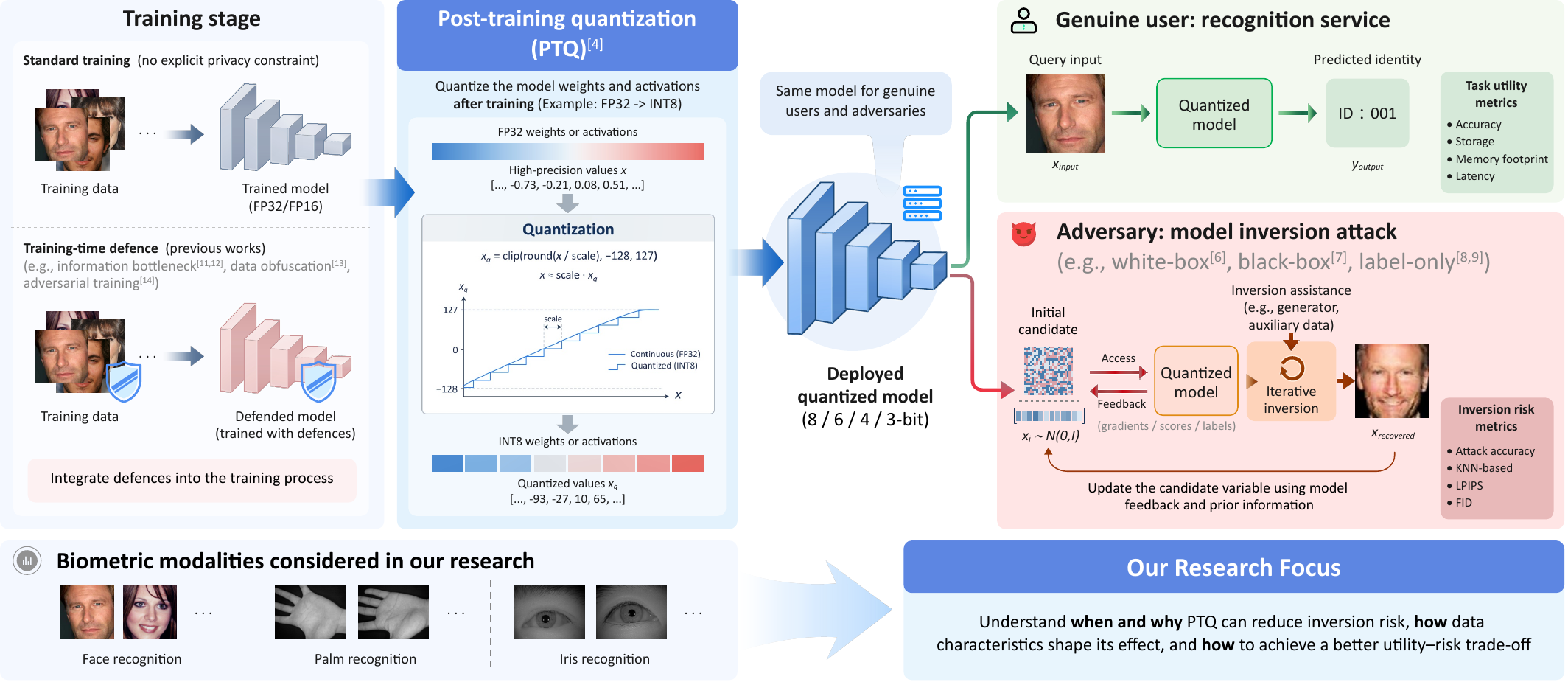}
\caption{Overview of model quantization and model inversion attacks in biometric recognition. Our research aims to understand how quantization and data characteristics affect inversion risk and to improve the trade-off between predictive utility and inversion risk.}
\label{fig:quantization_inversion_background}
\end{figure*}

\section{Introduction}
\label{sec:introduction}
\IEEEPARstart{M}{odel} quantization, a widely used model compression technique~\cite{Cheng_2018}, lowers the numerical precision of neural network weights and activations to reduce storage and computational costs~\cite{Jacob_2018_CVPR}.
As neural networks increase in scale and computational demand, quantization has become an important technique for deployment on resource-constrained devices.
Existing approaches broadly fall into quantization-aware training (QAT)~\cite{Jung_2019_CVPR}, which incorporates quantization effects during model training, and post-training quantization (PTQ)~\cite{pmlr-v119-nagel20a}, which operates on an already trained model.
In practice, deployment often begins with publicly available pretrained models, while access to the data and computational resources required for further training may be limited~\cite{Nagel_2019_ICCV}.
PTQ accommodates these constraints through calibration and, when necessary, optimization of quantization parameters, without retraining the underlying network weights.
Its broad applicability under such deployment constraints motivates our focus on PTQ.
As illustrated in Fig.~\ref{fig:quantization_inversion_background}, this process converts a pretrained full-precision model into a low-precision model.

Beyond its effects on storage and computation, quantization can alter a neural network's intermediate features and predictions~\cite{Nagel_2019_ICCV}, with potential implications for privacy.
Model inversion attacks (MIAs)~\cite{Yuan_2023,Han_2023_CVPR,Kahla_2022_CVPR,10458692} exploit information exposed through these features or outputs to recover or reconstruct sensitive content associated with training data or inference inputs.
In image recognition, such reconstructions can reveal sensitive identity information, as illustrated in Fig.~\ref{fig:quantization_inversion_background}.
Intuitively, quantization may influence what an adversary can recover from a deployed model.
Prominent existing defenses mitigate inversion risk through interventions during model training~\cite{Abadi_2016,Wang_2021,Peng_2022,Yi_2026,Xiao_2020,11537378}.
In practice, however, access to the complete original training dataset may be unavailable when deploying pretrained models, limiting the applicability of these training-time defenses~\cite{Yi_2026}.
These constraints also arise in PTQ deployment, raising a natural question: \textbf{\textit{How does model quantization affect model inversion risk?}}

However, the relationship between model quantization and MIAs remains insufficiently understood.
Existing studies have examined quantization in federated learning~\cite{Ovi_2023_CVPR} and collaborative inference~\cite{Niu_2024_CVPR,Chaopeng_2023}, as well as its effects on membership inference~\cite{317957,pmlr-v300-aubinais26a}, but address different quantization settings or privacy objectives.
In particular, defenses that quantize shared gradients to mitigate gradient inversion~\cite{Ovi_2023_CVPR} concern training-time observations and do not directly establish how PTQ affects information exposed through deployed-model predictions.
Moreover, reduced attack success may reflect changes in information availability, difficulties in attack optimization, or both.
For white-box MIAs, hard activation quantization may impede direct gradient propagation through the target model.
We therefore focus primarily on label-only attacks~\cite{Kahla_2022_CVPR,NEURIPS2023_d9827e81,10458692,10003239} and also consider black-box attacks using confidence vectors~\cite{Han_2023_CVPR}, neither of which requires target-model gradients.
We examine how quantization changes what an adversary can recover from model predictions and how accurately.

Models trained on data with different characteristics may learn distinct feature structures and respond differently to the same quantization scheme.
Since MIAs exploit these features and the resulting predictions~\cite{liu2026partitioning}, differences in quantization response may also influence inversion risk.
Face recognition has served as a common benchmark for generative image MIAs~\cite{Zhang_2020_CVPR,Chen_2021_ICCV,pmlr-v162-struppek22a}.
Although reconstruction attacks have also been studied for other biometric modalities~\cite{Yan_2024,Ahmad_2022}, how quantization affects inversion risk across modalities remains unclear.
Palmprints and irises, shown in Fig.~\ref{fig:quantization_inversion_background}, are likewise widely used biometric identifiers, yet their discriminative cues differ from those of facial images: directional line patterns in palmprints and dense, fine-grained textures in irises.
These structural differences may affect both baseline inversion difficulty and the sensitivity of learned features to quantization.
This raises a second question: \textbf{\textit{How do data characteristics shape the relationship between quantization and MIAs?}}

To address these questions, we combine information-theoretic analysis with empirical investigation.
For the first question, we trace the propagation of local quantization errors to bound changes in mutual information between inputs and a categorical random variable defined by the model's prediction probabilities.
We further analyze hard-label stability in terms of decision margins and characterize perturbations to representation geometry through feature statistics.
For the second question, we examine how within-class variation in inputs and learned representations relates to baseline inversion difficulty, distinguishing the recovery of representative images of a target identity from the reconstruction of specific inputs.
We then investigate how data texture characteristics relate to quantization-induced changes in feature distributions, model predictions, and inversion outcomes.
Experiments reveal pronounced differences in quantization sensitivity across biometric data domains and architectures at 4 bits.
These findings motivate allocating precision according to task sensitivity and constraining representation geometry during utility recovery.

Building on these insights, we propose a privacy-aware mixed-precision PTQ method to improve the trade-off between predictive utility and inversion resistance.
The method consists of 3 steps.
\textit{1) Task sensitivity: We estimate layer-wise task sensitivity at candidate bit widths.} A Fisher-type proxy~\cite{Fisher_1992} weights quantization errors to assess their potential impact on the recognition task.
\textit{2) Bit allocation and calibration: We construct an initial quantized model through budget-aware bit allocation and activation calibration.} The sensitivity estimates guide weight-bit allocation, while activation bit widths are adjusted according to a separate budget. Activation ranges are then calibrated to initialize the quantized model.
\textit{3) Quantizer refinement: We jointly optimize weight and activation quantization scales and weight-rounding decisions to balance utility recovery and inversion resistance.} The bit allocation and underlying network weights remain fixed.
Classification loss and distillation from a frozen full-precision teacher guide utility recovery.
To limit the restoration of class structure that may facilitate inversion, geometry regularization penalizes inter-class expansion in feature and output spaces relative to the initial quantized model.
Additional regularization limits scale drift and encourages soft rounding decisions to approach binary values.
By operating on quantization parameters, our method is orthogonal to output post-processing and can be combined with compatible output-level defenses.

Our evaluation covers recognition utility, prediction and representation changes, identity recovery, and quantitative and qualitative input reconstruction quality.
Quantization analyses cover FaceScrub~\cite{Ng_2014}, Palm\footnote{CASIA-Palmprint: \url{http://english.ia.cas.cn/db/201611/t20161101_169936.html}.}, and Iris-Thousand\footnote{CASIA-Iris-Thousand, a subset of CASIA-IrisV4: \url{http://english.ia.cas.cn/db/201610/t20161026_169399.html}.} using ResNet-50~\cite{He_2016_CVPR}, ViT~\cite{dosovitskiy2020vit}, and Swin Transformer~\cite{Liu_2021_ICCV}.
We evaluate label-only and confidence-vector black-box MIAs on ResNet-50 across the 3 datasets, with additional evaluations on Swin Transformer for iris recognition.
We also reconstruct individual inputs from intermediate features and predictions on FaceScrub.
Strict success requires recognition as the intended identity by both the target model and an independent evaluator; lower rates indicate fewer successful identity recoveries.
For ResNet-50, our Palm–4-bit model reduces RL-MIA's strict success~\cite{Han_2023_CVPR} from 54\% to 26\%, while accuracy decreases from 99.01\% to 96.55\%, relative to FP32.
Adding Stealthy Shield Defense (SSD, \(\epsilon=0.1\))~\cite{ICLR2025_1d8e261c} to Iris–4.5-bit lowers BREP-MI's strict success~\cite{Kahla_2022_CVPR} from 63.33\% to 37.33\% and accuracy from 92.8\% to 87.6\%, relative to quantization alone.

Our main contributions are summarized as follows:
\begin{itemize}
\item We relate quantization to MIAs by bounding mutual information changes between inputs and categorical variables defined by prediction probabilities, and analyzing hard-label stability and representation geometry \textit{(Sec.~\ref{sec:information_theory})}.
\item We analyze how within-class variation and texture characteristics relate to baseline inversion difficulty and quantization-induced changes in representations, predictions, and inversion outcomes \textit{(Secs.~\ref{sec:information_theory} and~\ref{sec:data_effects})}.
\item We propose privacy-aware mixed-precision PTQ with task-sensitive bit allocation and geometry-regularized scale and rounding optimization, balancing utility and inversion resistance \textit{(Sec.~\ref{sec:privacy_aware_ptq})}.
\item We introduce complementary evaluation metrics and conduct extensive experiments across datasets, architectures, and attack settings, including ablation studies, to validate our analyses and assess the method's utility--privacy trade-offs \textit{(Sec.~\ref{sec:experiments})}.
\end{itemize}

\section{Related Work}
\label{sec:related_work}
We review model quantization, MIAs and their defenses, and quantization-related privacy research, while clarifying the scope of our study.

\subsection{Model Quantization}
\label{subsec:model_quantization}
Model quantization approximates neural network weights and activations using a finite set of low-precision values, aiming to reduce storage and computational costs while preserving predictive performance~\cite{Jacob_2018_CVPR}.
QAT accounts for quantization effects during network training or fine-tuning~\cite{Jung_2019_CVPR}, whereas PTQ converts pretrained models through calibration and, when needed, optimization of quantization parameters~\cite{pmlr-v119-nagel20a}.
Independently of this distinction, uniform-precision methods use fixed bit widths across layers, while mixed-precision methods allocate different bit widths according to layer sensitivity~\cite{Dong_2019_ICCV}.
Although efficient 4-bit and sub-4-bit implementations have been demonstrated, their support remains tied to specific hardware architectures, numerical formats, and operator kernels~\cite{pmlr-v139-yao21a,Wei_2025}.
These constraints complicate efficient deployment of flexible layer-wise precision assignments across models and platforms, and lower bit widths do not necessarily yield proportional inference speedups~\cite{pmlr-v139-yao21a}.

In practical deployment, access to the original training data and resources for retraining a pretrained model may be limited~\cite{Nagel_2019_ICCV}.
PTQ offers flexibility under these constraints by separating quantization from the original training process, motivating our focus on this approach.
Existing PTQ methods mitigate quantization error through several strategies.
DFQ~\cite{Nagel_2019_ICCV} applies weight equalization and bias correction without requiring training data.
AdaRound~\cite{pmlr-v119-nagel20a} optimizes weight-rounding decisions using calibration data, while BRECQ~\cite{2021-iclr-brecq} reconstructs network blocks to account for dependencies between layers.
A uniform quantizer with zero point set to zero has the dequantized form~\cite{Jacob_2018_CVPR}
\begin{equation}
\begin{aligned}
\widehat v = \mathcal Q_b(v;s)= s\,\operatorname{clip}\!\left(
\operatorname{round}\!\left(\frac{v}{s}\right),
q_{\min}(b),q_{\max}(b)
\right),
\end{aligned}
\label{eq:uniform_quantization}
\end{equation}
where $v$ is a scalar weight or activation, $b$ is the bit width, $s>0$ is the quantization step size, and $q_{\min}(b)$ and $q_{\max}(b)$ define the integer code range.
The operators $\operatorname{round}$ and $\operatorname{clip}$ perform nearest-neighbor rounding and clipping, respectively.
For tensors, the operation is applied elementwise, with $s_W$ and $s_A$ denoting the weight and activation scales.

Our study examines how PTQ affects inversion risk across data domains and architectures, and uses these findings to guide privacy-aware quantization.
Accordingly, our evaluation focuses on assessing the proposed design against a calibration-based PTQ reference and through component ablations, rather than establishing superiority over alternative PTQ algorithms.
We therefore do not include additional PTQ methods primarily designed to optimize predictive utility as baselines.

\subsection{Model Inversion Attacks and Defenses}
\label{subsec:model_inversion}
Fredrikson et al.~\cite{184489} first introduced MIAs in pharmacogenetic prediction, inferring sensitive genetic attributes from trained models and auxiliary information.
Their subsequent work~\cite{Fredrikson_2015} exploited prediction confidence to recover recognizable facial images.
Subsequent image-based MIAs consider different access levels to a trained target model $f$.

White-box attacks have access to the target model's structure, parameters, and computation.
Generative white-box MIAs typically seek a representative image of a target identity $y$ by optimizing latent variables under a generative prior.
A representative objective is
\begin{equation}
z^\star\in\operatorname*{arg\,min}_{z}\left[\mathcal L_{\mathrm{id}}(G(z);f,y)+\lambda R(z)\right],
\label{eq:white_box_paradigm}
\end{equation}
where $G$ is a generator with identity conditioning implicit where applicable, $\mathcal L_{\mathrm{id}}$ an identity loss derived from $f$, and $R$ an optional regularizer with weight $\lambda\geq0$.
Gradient-based optimization through $f$ typically yields an approximate solution $\widehat z$ and the recovered image $\widehat x=G(\widehat z)$.
GMI~\cite{Zhang_2020_CVPR} learns a generative prior from public data and optimizes latent codes through the target model.
KED-MI~\cite{Chen_2021_ICCV} incorporates target predictions into the training of a generative adversarial network (GAN)~\cite{NIPS2014_f033ed80} and optimizes a latent distribution, while PPA~\cite{pmlr-v162-struppek22a} uses publicly available pretrained generators under distribution shifts.
PLG-MI~\cite{Yuan_2023} uses target-derived pseudo-labels to train a conditional GAN and organize inversion by class.
Hard rounding has zero derivatives almost everywhere, while native integer operators may lack backward support, obstructing direct-gradient white-box optimization.
As shown in Fig.~\ref{fig:plg_int8_gradient}, our PLG-MI tests on INT8-quantized Palm models exhibit no effective inversion gradient.
We therefore exclude these attacks from our large-scale cross-dataset evaluation.

% Place the INT8 gradient diagnostic at the top of the next available column.
\begin{figure}[!t]
\centering
\includegraphics[width=\columnwidth]{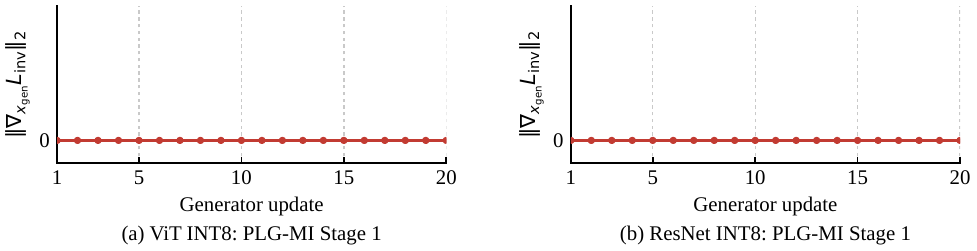}
\caption{\textbf{Blocked inversion gradients in the original PLG-MI attack~\cite{Yuan_2023} on INT8-quantized Palm classifiers.} The checkpoints were produced using PyTorch's official \href{https://docs.pytorch.org/docs/2.9/quantization-support.html}{static PTQ APIs}~\cite{NEURIPS2019_bdbca288}. Across 20 Stage-1 generator updates, the inversion-loss gradient with respect to the generated image was zero for (a) ViT and unavailable for (b) ResNet-50; unavailable gradients are plotted as zero. Separate Stage-2 first-step checks likewise found a zero latent gradient for ViT and no autograd path for ResNet-50.}
\label{fig:plg_int8_gradient}
\end{figure}

Confidence-based black-box attacks obtain prediction vectors without access to the target model's internal parameters.
Their query-driven learning or search can be summarized as
\begin{equation}
\widehat x=\mathcal A_{\mathrm{CB}}\!\left(c;\{(x_t,\mathbf p_f(x_t))\}_{t=1}^{T},\mathcal D_{\mathrm{aux}}\right),
\label{eq:confidence_paradigm}
\end{equation}
where $\mathbf p_f(x_t)$ is the class-probability vector returned for query $x_t$, $T$ the query count, and $\mathcal D_{\mathrm{aux}}$ permitted auxiliary data.
The attack condition $c$ specifies a target identity or a separately observed prediction of a particular unknown input.
The attack algorithm $\mathcal A_{\mathrm{CB}}$ may use these responses to learn an inverse mapping, evaluate search candidates, or construct rewards.
Queries may be adaptive, and attack algorithms may be randomized.
Yang et al.~\cite{Yang_2019} learn an inversion network from auxiliary inputs and target predictions, considering both reconstruction of specific inputs and recovery of representative samples of training classes.
We use reconstruction and recovery with these respective meanings throughout this paper.
The black-box variant of MIRROR~\cite{An.Mirror.NDSS.2022} performs confidence-guided genetic search in the latent space of StyleGAN~\cite{Karras_2019_CVPR}.
RL-MIA~\cite{Han_2023_CVPR} formulates latent-space exploration as a Markov decision process~\cite{Puterman_1990} and uses prediction confidence to construct rewards.
Both methods primarily seek representative images of target identities in the settings considered here.

Label-only MIAs observe only the predicted class $h_f(x)=\arg\max_k p_{f,k}(x)$, with a fixed rule for resolving ties.
For recovery of a target identity $y$, label-only access can be expressed as
\begin{equation}
\widehat x=\mathcal A_{\mathrm{LO}}\!\left(y;\{(x_t,h_f(x_t))\}_{t=1}^{T},\mathcal D_{\mathrm{aux}}\right).
\label{eq:label_only_paradigm}
\end{equation}
Here, $\mathcal A_{\mathrm{LO}}$ uses labels to guide search, learn surrogates, or train conditional generators, without target-model confidence scores.
BREP-MI~\cite{Kahla_2022_CVPR} uses label queries to guide latent-space search through boundary repulsion.
Zhu et al.~\cite{10003239} estimate confidence vectors from label responses using boundary-distance approximations and shadow models, then train an inversion network on these estimates.
LOKT~\cite{NEURIPS2023_d9827e81} transfers knowledge from target labels to surrogate models and subsequently performs white-box inversion against the surrogates.
Liu et al.~\cite{10458692} use predicted labels to train a conditional diffusion model for recovering representative samples of target identities.
These attacks do not require target-model gradients, although they may use gradients of surrogate models.

Face recognition has served as a prominent benchmark across these attack settings, although reconstruction attacks have also been studied for other biometric modalities~\cite{Yan_2024,Ahmad_2022}.
We include palmprint and iris recognition, as palmprint line patterns and fine-grained iris textures differ from facial cues and may affect quantization responses and inversion outcomes.
For input reconstruction, we separately assume access to input-specific prediction vectors or intermediate features, which are unavailable through the label-only interface.

Prominent defenses intervene during model training.
MID~\cite{Wang_2021} regularizes mutual information between inputs and predictions.
BiDO~\cite{Peng_2022} discourages dependence between inputs and learned representations while encouraging dependence between representations and labels.
Related approaches include record-level privacy protection through DP-SGD~\cite{Abadi_2016}, training-data augmentation~\cite{Yi_2026}, and adversarial representation learning against reconstruction~\cite{Xiao_2020}.
In collaborative inference, SiftFunnel~\cite{11537378} constrains input--feature dependence, inter-channel correlation, and feature magnitudes during training.
SSD~\cite{ICLR2025_1d8e261c} instead post-processes prediction probabilities without retraining the target model and can be combined with our PTQ method.
Our study starts from an already trained model and focuses on how PTQ affects its vulnerability to MIAs: quantization configurations and parameters may be adjusted, while the underlying network weights are not retrained.
Accordingly, integrating training-time defenses into the target model's training process falls outside our scope.

\subsection{Quantization and Privacy Risks}
\label{subsec:quantization_privacy}
Prior work has examined the privacy implications of model compression, including quantization.
CompLeak~\cite{317957} evaluates several compression techniques and shows that combining information from original and compressed models can strengthen membership inference.
Aubinais et al.~\cite{pmlr-v300-aubinais26a} derive an indicator of membership inference security for PTQ through asymptotic analysis, enabling comparisons between quantization procedures.
These studies demonstrate that quantization can affect privacy, but membership inference concerns whether a sample participated in training and does not directly address the recovery of sensitive content.

Prior work in federated learning and collaborative inference more directly addresses content recovery.
Ovi et al.~\cite{Ovi_2023_CVPR} apply mixed quantization to shared gradients during federated training to mitigate attacks that aim to reconstruct local training data from transmitted gradient observations.
Delta~\cite{Niu_2024_CVPR} retains the low-dimensional, information-sensitive component of intermediate representations in a private environment, while applying random binary quantization and a differential privacy mechanism to the residuals before transmission to a public environment.
An energy-aware DNN partitioning method~\cite{Chaopeng_2023} combines partition selection with mixed-precision model quantization under energy constraints, using discrepancies between input images and projected feature maps as a privacy proxy.

These studies do not yet provide a systematic account of how weight and activation PTQ affect inversion risk from the predictions of pretrained image models, particularly across data domains and architectures.
We investigate this relationship through theoretical analysis and empirical evaluation, using the resulting understanding to guide privacy-aware PTQ that balances predictive utility and inversion resistance.

\section{Information-Theoretic Relationship between Model Quantization and Model Inversion}
\label{sec:information_theory}

This section examines the relationship between model quantization and MIAs, including input reconstruction and identity recovery.
We develop an information-theoretic perspective and analyze changes in categorical prediction information, hard-label stability, and representation geometry.

\subsection{Information Recoverability under Model Quantization}
\label{subsec:information_recoverability}

MIAs seek to reconstruct a specific input or recover representative samples of a target identity.
To make the role of available information precise, we first consider reconstruction of an unknown input $X\sim P_X$.
Let $O$ collect the model observations and auxiliary information available to the attacker, including the full query--response transcript for adaptive attacks.
A possibly randomized attack produces $\widehat X$ from $O$, yielding the Markov chain $X\rightarrow O\rightarrow\widehat X$.

Let $d$ measure reconstruction distortion and $R_X(D)$ denote the corresponding rate--distortion function.
For any attack satisfying $\mathbb E[d(X,\widehat X)]\le D$, rate--distortion theory and the data processing inequality give~\cite{Cover_2005}
\begin{equation}
R_X(D)\le I(X;\widehat X)\le I(X;O).
\label{eq:inversion_information_requirement}
\end{equation}
The available information therefore constrains achievable reconstruction fidelity, while the attack determines how effectively that information is used.
This bound concerns individual-input reconstruction.
Identity recovery uses model information to generate representative samples without requiring proximity to a particular original input.

Applying a fixed quantizer $Q$ to an existing feature, gradient, or transmitted vector $V$ cannot increase its information about the input: $I(X;Q(V))\le I(X;V)$.
However, quantizing a model's weights and internal activations changes the computation that produces its predictions.
The quantized model's predictions need not be obtainable from the original predictions alone.
The data processing inequality therefore does not imply that model quantization reduces the information in predictions.

Accordingly, we compare both models under the same input distribution, auxiliary information, and access setting.
We then relate local quantization errors to categorical prediction information, hard-label stability, and representation geometry.

\subsection{Quantization-Induced Changes in Prediction Information}
\label{subsec:prediction_information}
We analyze prediction information through auxiliary categorical variables under the common input distribution $P_X$.
Let $f_q$ denote the model obtained from $f$ using a fixed quantization configuration $q$, which may assign different bit widths and scales to individual quantization sites.
Both models have $M\ge2$ classes; write $\mathbf p(x)=\mathbf p_f(x)$ and $\mathbf p_q(x)=\mathbf p_{f_q}(x)$.
Define the auxiliary variables by $C\mid X=x\sim\operatorname{Cat}(\mathbf p(x))$ and $C_q\mid X=x\sim\operatorname{Cat}(\mathbf p_q(x))$, and let $\Delta I_q=I(X;C_q)-I(X;C)$.
We reserve $Y$ for the ground-truth class.
Here, larger $I(X;C)$ indicates stronger statistical dependence between inputs and the auxiliary prediction class $C$.
The bounds concern $C,C_q$, not the full score vector or attacker view $O$.
All logarithms are natural, and information is measured in nats.

To connect prediction changes with information changes, define the pointwise total variation and its input average as
\begin{equation}
\delta_q(x)=\tfrac12\|\mathbf p_q(x)-\mathbf p(x)\|_1,
\qquad \varepsilon_q=\mathbb E_X\delta_q(X).
\label{eq:prediction_tv}
\end{equation}
Here, $\delta_q(x)=0$ means identical probability vectors at $x$; larger values indicate stronger perturbations.
We next use entropy continuity to translate these probability differences into bounds on entropy changes.
Let $h_2(t)=-t\log t-(1-t)\log(1-t)$, with $0\log0=0$.
For $t\in[0,1]$, introduce the nondecreasing, concave function
\begin{equation}
\Phi_M(t)=\begin{cases}
h_2(t)+t\log(M-1),&t\le1-M^{-1},\\
\log M,&t>1-M^{-1}.
\end{cases}
\label{eq:entropy_continuity_envelope}
\end{equation}
The Fannes--Audenaert inequality~\cite{Audenaert_2007} implies $|\mathrm H(r)-\mathrm H(s)|\le\Phi_M(\operatorname{TV}(r,s))$ for two distributions on the same $M$ classes.
The constant branch in~\eqref{eq:entropy_continuity_envelope} keeps $\Phi_M$ nondecreasing, allowing an upper bound on total variation to replace its exact value.
Its concavity allows Jensen's inequality to bound the expected entropy change in terms of $\varepsilon_q$.

\begin{lemma}[Prediction-information continuity]
\label{lem:prediction_continuity}
Under the shared input distribution and class space,
\begin{equation}
|\Delta I_q|\le\min\{\log M,\,2\Phi_M(\varepsilon_q)\}.
\label{eq:prediction_mi_continuity}
\end{equation}
\end{lemma}
\begin{IEEEproof}
Since $I(X;C)=\mathrm H(C)-\mathrm H(C\mid X)$, we bound the marginal and conditional entropy changes separately.
The marginal class distributions are $\bar{\mathbf p}=\mathbb E_X\mathbf p(X)$ and $\bar{\mathbf p}_q=\mathbb E_X\mathbf p_q(X)$.
The triangle inequality gives $\operatorname{TV}(\bar{\mathbf p}_q,\bar{\mathbf p})\le\varepsilon_q$, so their entropies differ by at most $\Phi_M(\varepsilon_q)$.
For conditional entropies, the bound and Jensen's inequality give
\begin{equation}
\begin{aligned}
|\mathrm H(C_q\mid X)-\mathrm H(C\mid X)|
\le\mathbb E_X\Phi_M(\delta_q(X))
\le\Phi_M(\varepsilon_q).
\end{aligned}
\label{eq:conditional_entropy_continuity}
\end{equation}
Combining the two entropy bounds yields the factor $2$ in~\eqref{eq:prediction_mi_continuity}.
The additional cap follows because both mutual informations lie in $[0,\log M]$.
Appendix~A of the supplementary material gives every step, including the properties of $\Phi_M$.
\end{IEEEproof}

\noindent\textbf{Key takeaway.} If quantization leaves prediction probabilities nearly unchanged, the corresponding categorical mutual information also remains nearly unchanged.

It remains to relate $\varepsilon_q$ to quantization inside the network.
We use the dequantized real-valued representation in~\eqref{eq:uniform_quantization}, and order the quantization sites as $\mathcal S=\{1,\ldots,N\}$.
At each site, $e_{k,q}(x)$ denotes the output difference introduced by quantizing that site while keeping its incoming computation identical.
This isolates the site's new error from previously accumulated perturbations.

For an activation site, let $u_k$ be its pre-quantization tensor along the quantized prefix, $\mathcal C_{k,q}$ its componentwise clipping to the representable interval, and $\boldsymbol s_{k,q}$ the positive quantization steps broadcast to its coordinates.
For a coordinate with step $s$, the interval is $[s q_{\min},s q_{\max}]$, consistently with~\eqref{eq:uniform_quantization}.
Writing $u_k^{c}=\mathcal C_{k,q}(u_k)$ separates the rounding and clipping errors:
\begin{equation}
e_{k,q}^{A}=\mathcal Q_{k,q}(u_k)-u_k
=[\mathcal Q_{k,q}(u_k)-u_k^{c}]+[u_k^{c}-u_k].
\label{eq:local_activation_decomposition}
\end{equation}
Nearest-neighbor rounding contributes at most half a step per coordinate.
With $\omega_{k,q}(x)=\|u_k(x)-u_k^{c}(x)\|_2$, the triangle inequality therefore gives
\begin{equation}
\|e_{k,q}^{A}(x)\|_2
\le\tfrac12\|\boldsymbol s_{k,q}\|_2+\omega_{k,q}(x).
\label{eq:local_activation_bound}
\end{equation}
For a tensorwise step $s_{k,q}$ and $d_k$ coordinates, the first term equals $s_{k,q}\sqrt{d_k}/2$.

For an affine module $W_ka+c_k$, let $\Delta W_{k,q}$ and $\Delta c_{k,q}$ denote its parameter changes.
Evaluating both versions on the same prefix activation $a_{k-1}(x)$ gives
\begin{equation}
e_{k,q}^{W}(x)=\Delta W_{k,q}a_{k-1}(x)+\Delta c_{k,q},
\label{eq:local_weight_error}
\end{equation}
\begin{equation}
\|e_{k,q}^{W}(x)\|_2
\le\|\Delta W_{k,q}\|_{\mathrm{op}}\|a_{k-1}(x)\|_2
+\|\Delta c_{k,q}\|_2.
\label{eq:local_weight_bound}
\end{equation}
This expression accommodates learned weight-rounding decisions without assuming nearest-neighbor rounding.
Both $e_{k,q}^{W}$ and $e_{k,q}^{A}$ measure a local output difference with identical incoming computation, so we denote either by $e_{k,q}$ and treat them as separate sites in computational order.
When activation quantization follows a weight-quantized operation, its error is evaluated at the output produced using quantized weights.
Appendix~B provides the coordinatewise derivation and scalar-step special case.

\begin{lemma}[Propagation of local quantization errors]
\label{lem:quantization_propagation}
Suppose the common unquantized suffix after each site $k$ is $L_k^\downarrow$-Lipschitz in the Euclidean norm over all relevant prefix states and their local perturbations.
Then
\begin{equation}
\|\mathbf p_q(x)-\mathbf p(x)\|_2
\le\sum_{k\in\mathcal S}L_k^\downarrow\|e_{k,q}(x)\|_2.
\label{eq:local_to_output}
\end{equation}
For residual or branched networks, unchanged branch states are held fixed in each comparison; the constants must be valid uniformly over these shared states.
\end{lemma}
\begin{IEEEproof}
Construct $N+1$ hybrid models whose first $j$ sites are quantized and whose remaining sites retain the original computation, and denote their probability outputs by $\mathbf p^{(j)}$.
Then $\mathbf p^{(0)}=\mathbf p$, $\mathbf p^{(N)}=\mathbf p_q$, and
\begin{equation}
\mathbf p_q(x)-\mathbf p(x)
=\sum_{j=1}^{N}[\mathbf p^{(j)}(x)-\mathbf p^{(j-1)}(x)].
\label{eq:hybrid_telescoping}
\end{equation}
Two adjacent hybrids have the same quantized prefix and differ only at site $j$.
Their local output difference is therefore exactly $e_{j,q}(x)$, and their shared suffix maps it to an output difference of norm at most $L_j^\downarrow\|e_{j,q}(x)\|_2$.
Applying the triangle inequality to~\eqref{eq:hybrid_telescoping} proves~\eqref{eq:local_to_output}.
The full state-based construction appears in Appendix~C.
\end{IEEEproof}

Assume the local errors have finite first moments.
Taking expectations in~\eqref{eq:local_to_output} and using $\|v\|_1\le\sqrt M\|v\|_2$ yields the output perturbation budget
\begin{equation}
\varepsilon_q\le\mathcal G_q,
\qquad
\mathcal G_q=\frac{\sqrt M}{2}
\sum_{k\in\mathcal S}L_k^\downarrow\mathbb E_X\|e_{k,q}(X)\|_2.
\label{eq:output_perturbation_budget}
\end{equation}
Local errors reflect quantization ranges, clipping tails, and tensor dimensions; their means depend on $P_X$.
The $L_k^\downarrow$ must be valid upper bounds.

\noindent\textbf{Key takeaway.} The bound on quantization-induced prediction changes depends on both the model's error propagation and the input distribution.

\begin{theorem}[End-to-end prediction-information bound]
\label{thm:quantization_information}
Under the common input distribution and class space, the propagation conditions of Lemma~\ref{lem:quantization_propagation}, and finite mean local errors,
\begin{equation}
|\Delta I_q|\le
\min\{\log M,\,2\Phi_M(\min\{1,\mathcal G_q\})\}.
\label{eq:main_information_bound}
\end{equation}
If $\mathcal G_q\le1-M^{-1}$, this gives the explicit bound
\begin{equation}
|\Delta I_q|\le
2[\mathcal G_q\log(M-1)+h_2(\mathcal G_q)].
\label{eq:main_information_bound_small}
\end{equation}
\end{theorem}
\begin{IEEEproof}
Since total variation is at most one,~\eqref{eq:output_perturbation_budget} implies $\varepsilon_q\le\min\{1,\mathcal G_q\}$.
Substituting this into Lemma~\ref{lem:prediction_continuity} is valid because $\Phi_M$ is nondecreasing.
Its first branch gives~\eqref{eq:main_information_bound_small} under the stated condition.
Appendix~D gives the full derivation.
\end{IEEEproof}

The general bound includes weight, activation, and mixed-precision quantization through their local errors.
For activation-only quantization with tensorwise nearest-neighbor steps,~\eqref{eq:local_activation_bound} separates the step-size and clipping contributions:
\begin{equation}
\mathcal G_q\le\mathcal E_q^{A}
=\frac{\sqrt M}{2}\sum_{k\in\mathcal S}L_k^\downarrow
\left[\frac{s_{k,q}\sqrt{d_k}}2+\mathbb E_X\omega_{k,q}(X)\right].
\label{eq:activation_output_budget}
\end{equation}
Replacing $\mathcal G_q$ by $\mathcal E_q^{A}$ in~\eqref{eq:main_information_bound} gives its explicit activation-quantization specialization.
For a bit-indexed family $q_b$, write $\varepsilon_b=\varepsilon_{q_b}$ and $\Delta I_b=\Delta I_{q_b}$.
If its layerwise steps and clipping residuals additionally admit fixed common-scale bounds, these terms can be collected into
\begin{equation}
\varepsilon_b\le E_b=K_\Delta\Delta_b+K_{\mathrm{clip}}\Omega_b.
\label{eq:compact_quantization_budget}
\end{equation}
Here $\Delta_b$ and $\Omega_b$ are reference step and clipping scales, and the fixed coefficients aggregate layer dimensions and propagation bounds.
For $E_b\le1-M^{-1}$,
\begin{equation}
|\Delta I_b|\le2[E_b\log(M-1)+h_2(E_b)].
\label{eq:compact_information_bound}
\end{equation}
Appendix~D states the additional factorization assumptions; they are unnecessary for the general result~\eqref{eq:main_information_bound}.

Figure~\ref{fig:quantization_information_bounds}(a) illustrates an idealized fixed-range setting with $\Delta_b=2/(2^b-1)$, $K_\Delta=1$, and no clipping.
The bound decays as $O(b2^{-b})$ at high bit widths (Appendix~D).

\noindent\textbf{Key takeaway.} In the illustrated setting, the theoretical upper bound on quantization-induced changes in categorical mutual information is smaller at 8 and 6 bits than at 4 bits.
This contrast motivates our initial focus on 4-bit quantization, whose privacy--utility trade-off we then examine empirically.

We complement these bounds with empirical estimates of the same categorical mutual information from Palm--ResNet-50 predictions (Appendix~E).
On test inputs, estimates change little at 8 and 6 bits, but at 4 bits decrease from $3.769$ nats for FP32 to $3.254$ under weight-only quantization and $2.656$ under joint weight--activation quantization.
Weight quantization need not reduce prediction information.

\begin{figure}[!t]
\centering
\includegraphics[width=\columnwidth]{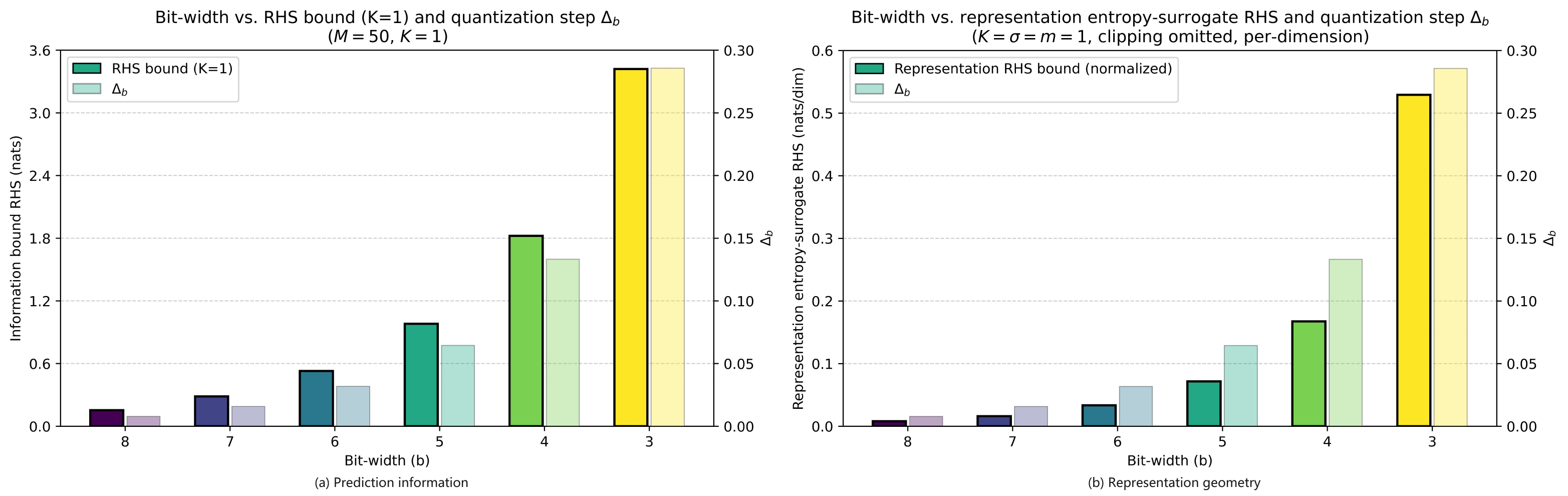}
\caption{\textbf{Illustrative bit-width dependence of information-related bounds.} Both panels use $\Delta_b=2/(2^b-1)$ for an idealized grid containing both endpoints of $[-1,1]$, with clipping omitted and propagation factors normalized. This grid differs from~\eqref{eq:uniform_quantization}. (a) The categorical information bound~\eqref{eq:compact_information_bound} with $M=50$ and $K_\Delta=1$. (b) The per-dimension form of~\eqref{eq:representation_entropy_bound}, $-\frac12\log[1-(2\Delta_b+\Delta_b^2)]$, with unit feature-spread and regularized covariance-floor parameters and representation error bounded by $\Delta_b$. The required condition $2\Delta_b+\Delta_b^2<1$ holds throughout. Full derivations appear in Appendices~D and~G.}
\label{fig:quantization_information_bounds}
\end{figure}

Label-only attacks observe the deterministic top class, which can remain unchanged despite changes in prediction probabilities.
Using the fixed tie-breaking rule in Sec.~\ref{subsec:model_inversion}, define $H=h_f(X)$, $H_q=h_{f_q}(X)$, and
\begin{equation}
\eta_q=\Pr[H_q\ne H].
\label{eq:label_disagreement}
\end{equation}
A larger $\eta_q$ means labels change more often after quantization, irrespective of their correctness.

\begin{corollary}[Hard-label information stability]
\label{cor:hard_label_information}
The information change at the hard-label interface satisfies
\begin{equation}
|I(X;H_q)-I(X;H)|\le\Phi_M(\eta_q).
\label{eq:hard_label_information_bound}
\end{equation}
\end{corollary}
\begin{IEEEproof}
Deterministic predictions satisfy $\mathrm H(H\mid X)=\mathrm H(H_q\mid X)=0$, so their mutual informations with $X$ equal their marginal entropies.
The two labels evaluated on the same input provide a coupling for which $\operatorname{TV}(P_H,P_{H_q})\le\Pr[H\ne H_q]=\eta_q$.
Applying the entropy-continuity bound proves~\eqref{eq:hard_label_information_bound}.
Only the marginal-entropy term remains, so the factor $2$ in~\eqref{eq:prediction_mi_continuity} is unnecessary here.
Appendix~F gives the full proof.
\end{IEEEproof}

To relate $\eta_q$ to quantization, let $m(x)=p_{(1)}(x)-p_{(2)}(x)$ be the difference between the largest and second-largest full-precision probabilities, and write $r_q(x)=\|\mathbf p_q(x)-\mathbf p(x)\|_\infty$.
If the original top class is unique and $r_q(x)<m(x)/2$, every competing class remains below it, since its probability gap after perturbation is at least $m(x)-2r_q(x)>0$.
Consequently,
\begin{equation}
h_{f_q}(x)\ne h_f(x)\quad\Longrightarrow\quad m(x)\le2r_q(x).
\label{eq:margin_label_stability}
\end{equation}
This implication also holds when $m(x)=0$, although it then gives no positive stability margin.
Failure of the sufficient condition $r_q(x)<m(x)/2$ does not imply a label flip.
For any $t>0$, a label flip lies in $\{m(X)\le2t\}$ or $\{r_q(X)\ge t\}$.
The union bound and Markov's inequality give
\begin{equation}
\begin{aligned}
\eta_q\le\Pr[m(X)\le2t]+\frac{\mathbb E r_q(X)}{t}
\le\Pr[m(X)\le2t]+\frac{\varepsilon_q}{t},
\end{aligned}
\label{eq:margin_disagreement_budget}
\end{equation}
The probability difference sums to zero, so $r_q(x)\le\delta_q(x)$; taking expectations gives the second inequality.
Combining~\eqref{eq:output_perturbation_budget} and~\eqref{eq:margin_disagreement_budget}, and optimizing over $t$, gives
\begin{equation}
\eta_q\le\bar\eta_q:=\min\left\{1,\inf_{t>0}\left[\Pr[m(X)\le2t]+\frac{\mathcal G_q}{t}\right]\right\}.
\label{eq:quantized_label_disagreement_bound}
\end{equation}
Monotonicity of $\Phi_M$ then gives $|I(X;H_q)-I(X;H)|\le\Phi_M(\bar\eta_q)$.
Appendix~F details the event inclusions and inequalities.
The bounds concern single predictions under $P_X$; adaptive multi-query MIAs are assessed experimentally.

\noindent\textbf{Key takeaway.} A hard label is preserved when every class probability changes by less than half the original prediction margin.
Across inputs, the information-stability bound depends on both prediction perturbations and the proportion of inputs with small prediction margins.

\subsection{Representation Geometry under Quantization}
\label{subsec:representation_geometry}
The preceding bounds characterize changes at the prediction interface.
We now examine intermediate representations through their covariance, which captures variation across inputs.
Since these representations need not share a finite alphabet, we use a regularized Gaussian entropy surrogate to quantify changes in their second-order geometry.

Fix a representation layer of dimension $d$ and configuration $q$.
Under the common input law $P_X$, let $Z=Z(X)$ and $Z_q=Z_q(X)$ be the full-precision and quantized representations in $\mathbb R^d$, both with finite second moments.
Define their difference and its centered scale by
\begin{equation}
E=Z_q-Z,\qquad
\tau^2=\mathbb E\|E-\mathbb EE\|_2^2.
\label{eq:representation_error_scale}
\end{equation}
Covariance is unchanged by a common shift of all feature vectors, so we center $E$ before measuring its effect.
Let $\mathcal S_Z$ contain all quantization sites used to compute the selected representation, including any quantizer applied at that output.
For each site, let $L_{k\to Z}$ be a valid Euclidean Lipschitz bound for the common suffix to the selected layer, with unchanged branch states held fixed as in Lemma~\ref{lem:quantization_propagation}. A quantizer at the selected output has an identity suffix.
The same hybrid construction gives $\|E(x)\|_2\le\sum_{k\in\mathcal S_Z}L_{k\to Z}\|e_{k,q}(x)\|_2$.
If the local errors have finite second moments, centering and the $L^2$ triangle inequality yield
\begin{equation}
\tau\le\bigl(\mathbb E\|E\|_2^2\bigr)^{1/2}
\le\sum_{k\in\mathcal S_Z}L_{k\to Z}
\bigl(\mathbb E\|e_{k,q}(X)\|_2^2\bigr)^{1/2}.
\label{eq:representation_error_propagation}
\end{equation}
A few large errors can contribute strongly to second moments even when most inputs incur little error.
To bound covariance changes, we therefore use the second moments in~\eqref{eq:representation_error_propagation}; the first-moment quantities in~\eqref{eq:output_perturbation_budget} alone do not suffice.

Let $\Sigma=\operatorname{Cov}(Z)$, $\Sigma_q=\operatorname{Cov}(Z_q)$, $\Sigma_E=\operatorname{Cov}(E)$, and $\Gamma_{ZE}=\operatorname{Cov}(Z,E)$.
The identity $Z_q=Z+E$ yields
\begin{equation}
\Delta\Sigma:=\Sigma_q-\Sigma=\Sigma_E+\Gamma_{ZE}+\Gamma_{ZE}^\top.
\label{eq:representation_covariance_decomposition}
\end{equation}
Set $\sigma=\sqrt{\|\Sigma\|_{\mathrm{op}}}$ as the spectral feature scale.
The error covariance satisfies $\|\Sigma_E\|_{\mathrm{op}}\le\operatorname{tr}(\Sigma_E)=\tau^2$.
For unit vectors $u,v$, Cauchy--Schwarz gives $|u^\top \Gamma_{ZE} v|\le\sqrt{u^\top\Sigma u}\sqrt{v^\top\Sigma_Ev}\le\sigma\tau$, hence $\|\Gamma_{ZE}\|_{\mathrm{op}}\le\sigma\tau$.
Applying the operator-norm triangle inequality to~\eqref{eq:representation_covariance_decomposition} therefore yields
\begin{equation}
\|\Delta\Sigma\|_{\mathrm{op}}\le2\sigma\tau+\tau^2.
\label{eq:representation_covariance_bound}
\end{equation}
Along a feature direction, the cross terms measure whether centered errors reinforce or oppose the original deviations from the mean.
They can therefore add to or offset the variance contributed by the errors themselves.
Neither independence nor a zero error mean is assumed.
Appendix~G, Steps~1--4, gives the full derivation.

\noindent\textbf{Key takeaway.} Quantization errors of the same magnitude can either increase or decrease variance along a feature direction.

To summarize covariance changes with a scalar, fix the same $\lambda>0$ for both models and define
\begin{equation}
\widetilde h_\lambda(Z)
=\frac12\log\!\left[(2\pi e)^d\det(\Sigma+\lambda I_d)\right],
\label{eq:representation_entropy_surrogate}
\end{equation}
with $\widetilde h_\lambda(Z_q)$ defined using $\Sigma_q$.
Here $I_d$ is the identity matrix.
This is the Gaussian differential-entropy formula~\cite{Cover_2005} evaluated at a regularized covariance, without assuming that the features themselves are Gaussian.
The determinant multiplies the regularized variances along principal directions, so larger surrogate values indicate greater overall covariance spread.
We next compare the covariance perturbation with the smallest regularized variance, defining
\begin{equation}
 m_\lambda=\lambda_{\min}(\Sigma+\lambda I_d),\qquad
 \chi=\frac{2\sigma\tau+\tau^2}{m_\lambda}.
\label{eq:representation_relative_budget}
\end{equation}
The regularization ensures $m_\lambda\ge\lambda>0$.

\begin{samepage}
\begin{proposition}[Representation entropy-surrogate stability]
\label{prop:representation_geometry}
For the fixed representations under the common input law, assume finite second moments and the same regularization parameter $\lambda>0$.
If $\chi<1$, then
\begin{equation}
\left|\widetilde h_\lambda(Z_q)-\widetilde h_\lambda(Z)\right|
\le\frac d2[-\log(1-\chi)].
\label{eq:representation_entropy_bound}
\end{equation}
\end{proposition}
\end{samepage}
\begin{IEEEproof}
Set $A=\Sigma+\lambda I_d$ and $B=A^{-1/2}\Delta\Sigma A^{-1/2}$.
Since $\|A^{-1/2}\|_{\mathrm{op}}^2=m_\lambda^{-1}$, submultiplicativity and~\eqref{eq:representation_covariance_bound} give
\begin{equation}
\|B\|_{\mathrm{op}}
\le\frac{\|\Delta\Sigma\|_{\mathrm{op}}}{m_\lambda}
\le\chi.
\label{eq:representation_relative_norm}
\end{equation}
The symmetric matrix $B$ has eigenvalues $\beta_i\in[-\chi,\chi]$.
The factorization $A+\Delta\Sigma=A^{1/2}(I_d+B)A^{1/2}$ implies
\begin{equation}
\widetilde h_\lambda(Z_q)-\widetilde h_\lambda(Z)
=\frac12\sum_{i=1}^{d}\log(1+\beta_i).
\label{eq:representation_spectral_difference}
\end{equation}
Because $\chi<1$, all logarithms are defined, and $1+\chi\le(1-\chi)^{-1}$ gives $|\log(1+\beta_i)|\le-\log(1-\chi)$.
Taking absolute values and summing proves~\eqref{eq:representation_entropy_bound}.
Appendix~G, Steps~5--8, details the matrix factorization and both eigenvalue cases.
\end{IEEEproof}

The transformation $A^{-1/2}$ measures covariance changes relative to the regularized variance along each original principal direction.
For fixed $\sigma$ and $\tau>0$, a smaller $m_\lambda$ weakens the bound within $\chi<1$.
Both regularized covariances remain positive definite even when this sufficient condition fails.
Comparisons use a common $\lambda$ and compatible normalization.

Figure~\ref{fig:quantization_information_bounds}(b) illustrates~\eqref{eq:representation_entropy_bound} under an idealized parameter setting.
With $\sigma=m_\lambda=1$, centered error scale $\tau_b\le\Delta_b=2/(2^b-1)$, and clipping omitted, the per-dimension bound is $-\tfrac12\log[1-(2\Delta_b+\Delta_b^2)]$ when $2\Delta_b+\Delta_b^2<1$.
Appendix~G, Steps~9--10, verifies this setting and the plotted values.

This surrogate describes internal geometry without assuming attacker access to features; it neither measures information in $O$ nor directly bounds reconstruction error or MIA success.

\noindent\textbf{Key takeaway.} The same nonzero variance change is proportionally larger along a direction with smaller regularized variance.
The stability bound depends on quantization error relative to the original feature spread.

Both local error moments and representation covariance depend on the input distribution.
This motivates examining how data-dependent geometry and quantization responses relate to inversion outcomes across domains.

\section{Data Effects on Model Inversion under Quantization}
\label{sec:data_effects}
Section~\ref{sec:information_theory} links quantization effects to data-dependent error propagation and representation geometry.
To distinguish pre-existing inversion risk from quantization-induced changes, we first examine baseline inversion difficulty across data domains, then investigate how quantization changes model behavior and inversion outcomes (Fig.~\ref{fig:data_characteristics_quantization}).
Dataset and experimental details are provided in Sec.~\ref{subsec:experimental_setup}.

% Use the user's four-panel export; the wide float is placed on page 9.
\begin{figure*}[!t]
\centering
\includegraphics[width=\textwidth]{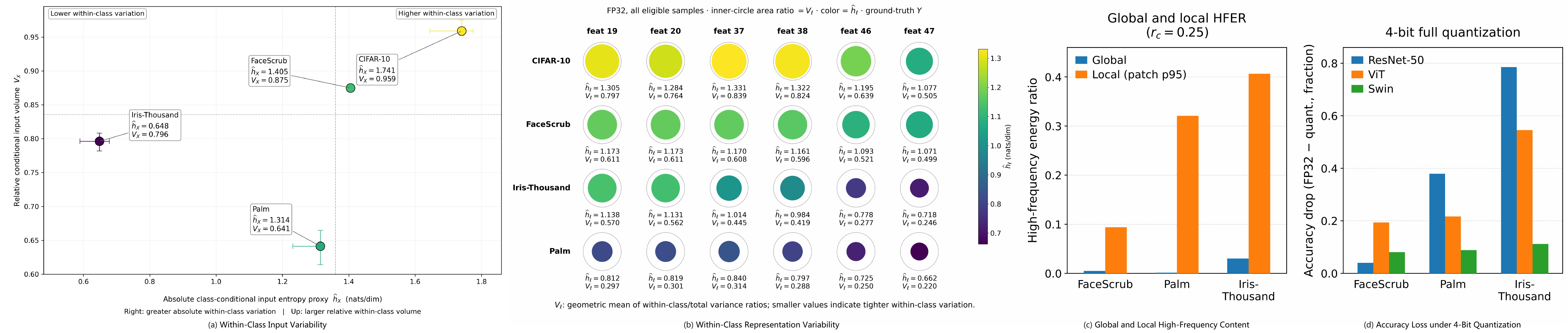}
\caption{\textbf{Data characteristics and model responses to quantization.}
(a) Input variability in a shared 32-dimensional PCA space~\cite{Pearson_1901}: the horizontal axis shows an absolute Gaussian within-class entropy proxy, and the vertical axis shows relative conditional volume. Error bars denote 95\% class-bootstrap confidence intervals.
(b) Within-class variability of FP32 representations across layers. Inner-circle area relative to the outer circle represents the geometric mean of per-coordinate within-class-to-total variance ratios; color indicates the corresponding normalized Gaussian entropy proxy. Conditioning in (a) and (b) uses the ground-truth class label $Y$.
(c) Dataset means of the global high-frequency energy ratio (HFER) and the per-image patch-level 95th-percentile HFER, with radial cutoff $r_c=0.25$, based on spectral texture analysis~\cite{Haralick_1979,Wang_2020_CVPR}.
(d) FP32-minus-quantized accuracy under 4-bit weight-and-activation quantization across datasets and architectures.
Panels use the analysis populations specified in Appendices~H and~I; comparisons across panels are at the domain level.}
\label{fig:data_characteristics_quantization}
\end{figure*}

\subsection{Data-Dependent Model Inversion Difficulty}
\label{subsec:data_dependent_inversion}
To relate data structure to reconstruction and identity recovery~\cite{11537378}, we separate the global covariance in Sec.~\ref{subsec:representation_geometry} into class-wise components.
Under a class-balanced input law, $\Sigma_T=\operatorname{Cov}(Z)$ satisfies
\begin{equation}
\Sigma_T=\Sigma_W+\Sigma_B,
\label{eq:class_covariance_decomposition}
\end{equation}
where $\Sigma_W=\mathbb E_Y[\operatorname{Cov}(Z\mid Y)]$ captures within-class variation and $\Sigma_B=\operatorname{Cov}_Y(\mathbb E[Z\mid Y])$ captures variation between class means.
Reconstructing a particular input requires instance-specific information, whereas recovering a representative identity need not preserve all within-class differences.
Appendix~\mbox  {H-E} connects this decomposition to the class-balanced estimators below.

Figure~\ref{fig:data_characteristics_quantization}(a) summarizes within-class input variation in a shared principal component analysis (PCA) space~\cite{Pearson_1901}.
The Gaussian entropy proxy $\widehat h_X$ summarizes absolute within-class spread in these shared coordinates, with higher values indicating greater spread.
The relative conditional-volume proxy $V_X$ summarizes within-class variation relative to total variation: lower values indicate tighter same-identity geometry relative to the overall spread, and higher values indicate a greater within-class share.
Palm has the lowest $V_X$ ($0.641$), compared with FaceScrub ($0.875$) and Iris-Thousand ($0.796$).
In contrast, Iris-Thousand has the lowest $\widehat h_X$, $0.648$ nats per dimension, compared with $1.405$ for FaceScrub.
Palm's localized line patterns may make identity recovery particularly sensitive to the generative prior's coverage of identity-specific texture.
Together with its relatively concentrated class-conditional geometry, this motivates our hypothesis that current generative priors recover Palm identities less readily than those of FaceScrub or Iris-Thousand.
Appendix~\mbox  {H-A} details the estimators.
Both orderings persist for PCA dimensions $16$, $32$, and $64$ and with common ridge (Appendix~\mbox  {H-C}).

Input structure alone does not characterize the target model.
Figure~\ref{fig:data_characteristics_quantization}(b) therefore examines relative within-class variation in global-average-pooled FP32 representations.
Its relative proxy $V_\ell$ is the geometric mean of coordinatewise within-class-to-total variance ratios, rather than a full-covariance measure as in the input analysis; their values are not directly comparable.
At the deepest monitored layer, $V_\ell$ is $0.499$ for FaceScrub, $0.220$ for Palm, and $0.246$ for Iris-Thousand.
The lower Palm and Iris-Thousand values indicate more concentrated within-identity features relative to the overall spread.
Such concentration can coexist with clear differences between identities, which identity-recovery attacks seek to exploit.
Palm already exhibits relatively concentrated within-class features at earlier layers, whereas Iris-Thousand undergoes a larger reduction across the monitored layers: $56.9\%$, compared with $18.3\%$ for FaceScrub.
Among correctly classified samples, FaceScrub and CIFAR-10 retain higher values than Palm and Iris-Thousand (Appendices~\mbox  {H-B} and~\mbox  {H-C}).

The FP32 attacks reveal corresponding differences in baseline identity recovery.
BREP-MI~\cite{Kahla_2022_CVPR} achieves strict success rates of $20.00\%$ on Palm, $52.08\%$ on FaceScrub, and $66.00\%$ on Iris-Thousand; RL-MIA~\cite{Han_2023_CVPR} achieves $54.00\%$, $57.92\%$, and $90.67\%$, respectively.
Palm is thus harder to recover in these settings, particularly for BREP-MI, while Iris-Thousand remains highly vulnerable despite its lower deep-layer within-class variation.
The strength of the domain difference depends on the attack: LOKT gives nearly equal success on Palm and FaceScrub.
Table~XII in Appendix~\mbox  {H-D} reports the complete FP32 comparisons.

Class cardinality, auxiliary-data coverage, generative priors, and model utility vary across domains (Appendix~\mbox  {H-D}).
These results establish each domain's starting risk under its attack setup.
The class-wise decomposition also identifies the geometry controlled by our regularizer: mean squared class-center distance is proportional to between-class covariance trace in the corresponding feature space (Appendix~\mbox  {H-E}).
Utility-oriented refinement can increase class separation to improve recognition.
Since identity-recovery attacks exploit class-discriminative model responses, restoring this separation may also restore useful attack cues.
We therefore penalize inter-class expansion relative to the initial quantized model, limiting this restoration while the task objective recovers utility (Sec.~\ref{sec:privacy_aware_ptq}).
We next assess quantization-induced changes against these domain-specific FP32 baselines.

\begin{samepage}
\noindent\textbf{Key takeaway.} Data characteristics and attack conditions jointly shape baseline inversion difficulty, providing the reference for assessing quantization-induced changes in risk.
\end{samepage}

\subsection{Data-Dependent Effects of Model Quantization}
\label{subsec:data_dependent_quantization}

The data-dependent error terms in Sec.~\ref{sec:information_theory} suggest that equal bit widths need not produce equal changes across domains.
Following spectral texture analysis~\cite{Haralick_1979,Wang_2020_CVPR}, we use the high-frequency energy ratio (HFER) to measure the fraction of Fourier power above a fixed radial cutoff.
Higher HFER indicates a greater relative high-frequency contribution; its patch-level 95th percentile highlights local regions with high ratios.
Figure~\ref{fig:data_characteristics_quantization}(c) reveals distinct texture profiles: Palm has a lower global HFER than FaceScrub but a higher local HFER, whereas Iris-Thousand has higher values of both.
The mean patch-level 95th-percentile HFER is $0.094$, $0.321$, and $0.406$ for FaceScrub, Palm, and Iris-Thousand, respectively.
Both global and local orderings persist after excluding auxiliary images (Appendix~\mbox  {I-A}).
These observations motivate the hypothesis that models relying on localized palmprint patterns or fine iris textures are more sensitive to low-precision quantization than face models, with the response also depending on the architecture.

% Introduce the method overview here so that the wide float appears at the top of page 10.
\begin{figure*}[!t]
\centering
\includegraphics[width=\textwidth]{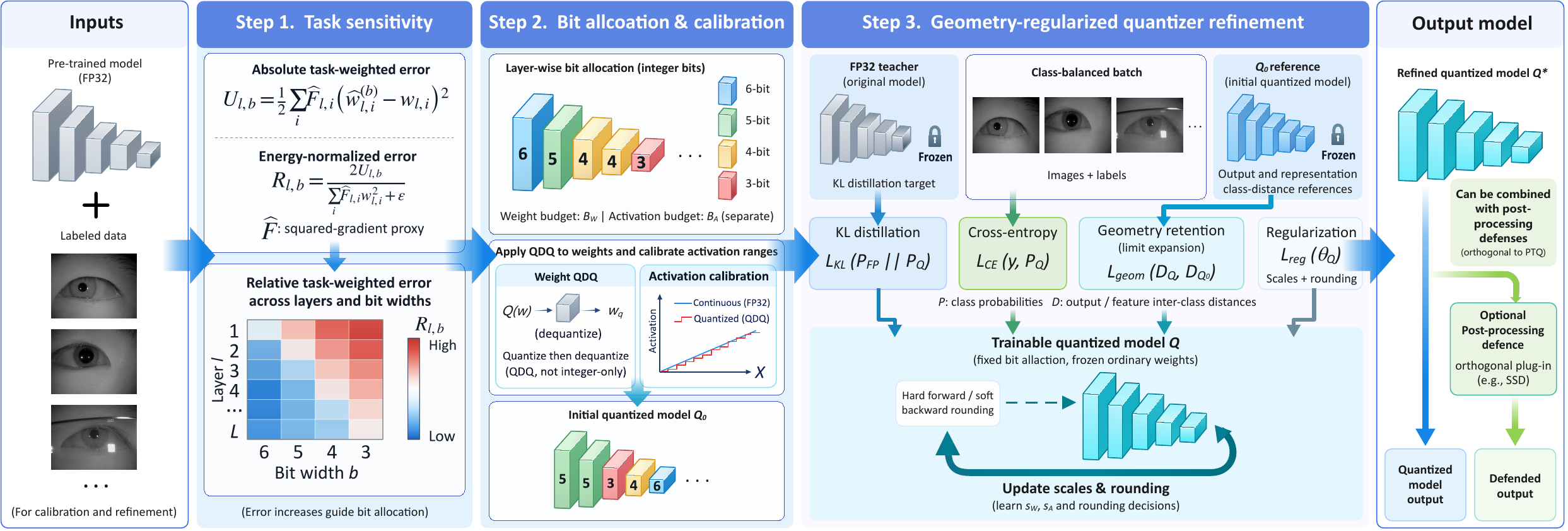}
\caption{\textbf{Overview of privacy-aware mixed-precision PTQ.} (1) Estimate task sensitivity at candidate bit widths. (2) Allocate weight and activation precision under separate budgets and calibrate activation ranges. (3) Refine quantization scales and weight-rounding decisions with fixed bit widths and network weights. The FP32 teacher guides utility recovery, while the initial quantized model provides the reference for limiting inter-class expansion.}
\label{fig:privacy_aware_ptq}
\end{figure*}

The 4-bit weight-and-activation results in Fig.~\ref{fig:data_characteristics_quantization}(d) support this hypothesis within the tested settings.
For ResNet-50, accuracy changes from $78.40\%$ to $74.40\%$ on FaceScrub, from $99.01\%$ to $61.08\%$ on Palm, and from $96.67\%$ to $18.13\%$ on Iris-Thousand.
The accuracy loss follows the same domain ordering for ViT and Swin, although the Palm--FaceScrub difference becomes much smaller.
Across the evaluated configurations, 8- and 6-bit quantization largely retain task accuracy, while 3-bit quantization often incurs severe degradation.
This contrast complements the theoretical motivation for studying 4 bits (Sec.~\ref{subsec:prediction_information}); Appendix~\mbox  {I-C} reports all configurations.

To trace the model response beyond accuracy, we examine both internal representations and predictions.
Paired feature cosine similarity compares FP32 and quantized features for the same input; values closer to one indicate better-preserved feature directions.
At 4 bits, the mean paired feature cosine similarity for ResNet-50 is $0.903$ on FaceScrub, compared with $0.724$ on Palm and $0.677$ on Iris-Thousand.
The corresponding label-disagreement rates are $17.60\%$, $38.92\%$, and $82.27\%$.
This disagreement directly estimates $\eta_q$ in Sec.~\ref{subsec:prediction_information}, connecting the empirical responses to the hard-label interface.
However, Palm has less label disagreement than FaceScrub under ViT and Swin despite a larger accuracy loss.
Quantization therefore changes representation alignment, prediction consistency, and task accuracy in related but distinct ways.
Appendices~\mbox  {I-B} and~\mbox  {I-C} give metric definitions, layerwise results, and weight-only comparisons.

To assess identity recovery, we use uniform calibration-only ResNet-50 PTQ models distinct from those in Fig.~\ref{fig:data_characteristics_quantization}(d), comparing each attack against its own FP32 baseline (Appendix~\mbox  {I-D}).
Their utility values therefore need not match the diagnostics.
At 4 bits, RL-MIA~\cite{Han_2023_CVPR} strict success decreases from $57.92\%$ to $53.58\%$ on FaceScrub and from $90.67\%$ to $34.00\%$ on Iris-Thousand.
Their task accuracies change from $78.98\%$ to $74.40\%$ and from $95.20\%$ to $21.33\%$, respectively.
The larger reduction on Iris is therefore accompanied by a substantial loss of utility.
The effect also depends on the attack: on Palm, RL-MIA strict success decreases from $54.00\%$ to $32.00\%$, whereas BREP-MI~\cite{Kahla_2022_CVPR} remains at $20.00\%$.
Model perturbation alone consequently does not specify the reduction achieved against a particular MIA.

Task accuracy is likewise insufficient to summarize inversion risk.
Palm's 8- and 6-bit PTQ models both achieve $98.03\%$ accuracy, yet RL-MIA strict success changes from $40.00\%$ to $28.00\%$ and BREP-MI from $30.00\%$ to $18.00\%$.
Reducing precision further to 4 bits raises these success rates to $32.00\%$ and $20.00\%$, even as accuracy falls to $60.10\%$.
These comparisons show why neither lower precision nor greater utility loss guarantees a better privacy--utility trade-off.

These differences across domains and models motivate task-sensitive bit allocation.
Subsequent utility recovery must also account for identity structure that may become easier to exploit.
Together with the class-wise geometry analysis in Sec.~\ref{subsec:data_dependent_inversion}, these findings motivate the constrained quantizer refinement developed in Sec.~\ref{sec:privacy_aware_ptq}.
Section~\ref{sec:experiments} evaluates the resulting trade-off through identity recovery and ablation studies.

\begin{samepage}
\noindent\textbf{Key takeaway.} Quantization changes inversion risk differently across data domains and attacks; useful protection requires assessing these changes jointly with task accuracy.
\end{samepage}

\section{Privacy-Aware Mixed-Precision PTQ}
\label{sec:privacy_aware_ptq}
The analyses in Secs.~\ref{sec:information_theory} and~\ref{sec:data_effects} motivate adapting precision to task sensitivity and constraining inter-class expansion during utility recovery.
Given a pretrained model, labeled calibration data, and separate weight and activation bit budgets, our method comprises three steps (Fig.~\ref{fig:privacy_aware_ptq}): task-sensitivity estimation, bit allocation and calibration, and geometry-regularized quantizer refinement.

\subsection{Task-Sensitivity Estimation}
\label{subsec:task_sensitivity}

Section~\ref{subsec:prediction_information} relates prediction changes to local quantization errors and their propagation, while Sec.~\ref{subsec:data_dependent_quantization} reveals different utility losses at equal precision across domains and architectures.
Inspired by Fisher-information analyses of layer importance for classification and MIAs~\cite{Ho_2024_CVPR} and of data leakage~\cite{pmlr-v161-hannun21a}, we use a diagonal Fisher-type proxy to assess the task impact of candidate quantization errors.

Let $W_\ell$ denote a layer's weight tensor and $\mathcal D_{\mathrm{cal}}$ calibration data labeled for the target task.
For mini-batch $\mathcal B_t\subseteq\mathcal D_{\mathrm{cal}}$, define the mean classification loss
\begin{equation}
\mathcal L_t
=-\frac{1}{|\mathcal B_t|}
\sum_{(x,y)\in\mathcal B_t}\log p_y(x),
\label{eq:calibration_task_loss}
\end{equation}
where $p_y(x)$ is the pretrained model's probability for class $y$.
With $g_{\ell,t}=\nabla_{W_\ell}\mathcal L_t$, we average the elementwise squared gradients over $T$ calibration mini-batches:
\begin{equation}
\widehat F_\ell
=\frac1T\sum_{t=1}^{T}g_{\ell,t}\odot g_{\ell,t}.
\label{eq:task_fisher_proxy}
\end{equation}
Larger entries indicate greater observed task-loss sensitivity.

For each candidate bit width $b$, obtain $\widehat W_\ell^{(b)}$ using~\eqref{eq:uniform_quantization}, with per-output-channel scales determined by maximum absolute weights and bit width.
Writing $\Delta W_{\ell,b}=\widehat W_\ell^{(b)}-W_\ell$, define the task-weighted error and its energy-normalized counterpart:
\begin{equation}
U_{\ell,b}=\tfrac12\sum_i\widehat F_{\ell,i}(\Delta W_{\ell,b,i})^2,
\;
R_{\ell,b}=\tfrac{2U_{\ell,b}}{\sum_i\widehat F_{\ell,i}W_{\ell,i}^2+\epsilon_{\mathrm{num}}},
\label{eq:task_weighted_quantization_error}
\end{equation}
Here $i$ indexes weight entries and $\epsilon_{\mathrm{num}}>0$ is the numerical stabilizer, shown as $\epsilon$ in Fig.~\ref{fig:privacy_aware_ptq}.
The quadratic form is a local sensitivity proxy, not an exact second-order loss change. Larger $U_{\ell,b}$ indicates greater estimated task cost; $R_{\ell,b}$ normalizes it by the layer's weighted parameter energy.
Reusing $\widehat F_\ell$ across bit widths, we guide budget-aware weight allocation to limit utility loss at low precision (Sec.~\ref{subsec:bit_allocation_calibration}).

\subsection{Budget-Aware Bit Allocation and Calibration}
\label{subsec:bit_allocation_calibration}

We prescribe target average bit widths $B_W$ and $B_A$ for weights and activations.
To approach these targets through mixed precision, we use the cost tables from Sec.~\ref{subsec:task_sensitivity} to assign weight bits $b_\ell^W$ and activation bits $b_j^A$ from integer candidates.
With $n_\ell$ quantized weight entries in layer $\ell$ and $m_j$ activation elements per input at site $j$, the average precisions are
\begin{equation}
\bar b_W=\frac{\sum_\ell n_\ell b_\ell^W}{\sum_\ell n_\ell},
\qquad
\bar b_A=\frac{\sum_j m_j b_j^A}{\sum_j m_j}.
\label{eq:average_bit_budgets}
\end{equation}

Weight allocation starts from the largest candidate bit width.
For a reduction $b\to b'$ to the next lower candidate in layer $\ell$, define $\Delta U_\ell=[U_{\ell,b'}-U_{\ell,b}]_+$ and $\Delta R_\ell=[R_{\ell,b'}-R_{\ell,b}]_+$, where $[z]_+=\max(z,0)$.
The reduction in average weight precision is $d_\ell=n_\ell(b-b')/\sum_k n_k$.
We score each move by
\begin{equation}
S_\ell(b\to b')
=\frac{\dfrac{\Delta U_\ell}{U_{\ell,b}+\epsilon_{\mathrm{num}}}
+\dfrac{\Delta R_\ell}{R_{\ell,b}+\epsilon_{\mathrm{num}}}}
{\max(d_\ell,\epsilon_{\mathrm{num}})}.
\label{eq:bit_allocation_score}
\end{equation}
Within each layer, $R_{\ell,b}$ is a fixed multiple of $U_{\ell,b}$; their relative-increase terms differ only through numerical stabilization.
We apply the reduction with the smallest score, repeating until $\bar b_W\le B_W$ or no further reduction is available.
This favors low relative cost increases per bit saved.

To reuse the weight-layer sensitivity estimates without constructing a separate activation cost table, we initialize activation bits from associated weight layers.
Let $\pi(j)$ map activation site $j$ to a weight layer according to the architecture; initially, $b_j^A=b_{\pi(j)}^W$, restricted to the allowed candidates.
We adjust one site at a time toward $B_A$, selecting the adjacent candidate that minimizes the resulting absolute gap $|\bar b_A-B_A|$.
We reject moves that increase this gap and stop at the target, when no admissible move remains, or at an iteration limit.
Discrete assignments can leave a residual gap to $B_A$.

The quantization-step and clipping terms in~\eqref{eq:local_activation_bound} motivate calibrating activation ranges after bit allocation.
With weights quantized, we collect channelwise extrema over $\mathcal D_{\mathrm{cal}}$, bypassing activation quantization during this pass.
The configured sign convention at site $j$ sets $(q_{\min,j},q_{\max,j})$ to $(0,2^{b_j^A}-1)$ for unsigned quantization or $(-2^{b_j^A-1},2^{b_j^A-1}-1)$ for signed quantization.
For channel $k$ with extrema $a_{jk}^{-},a_{jk}^{+}$, set $r_{jk}=\max(a_{jk}^{+},\epsilon_0)$ at unsigned sites and $r_{jk}=\max(|a_{jk}^{-}|,|a_{jk}^{+}|,\epsilon_0)$ at signed sites, where $\epsilon_0>0$ handles zero ranges.
These ranges define the nominal calibration step
\begin{equation}
s_{jk}^{*}
=\operatorname{clip}\!\left(\frac{r_{jk}}{q_{\max,j}},s_{\min},s_{\max}\right),
\label{eq:nominal_activation_scale}
\end{equation}
where $0<s_{\min}<s_{\max}$ are numerical scale bounds.
The trainable parameter $\xi_{jk}^A$ determines the effective step through
\begin{equation}
s_{jk}^{A}
=\operatorname{clip}\!\left(\operatorname{softplus}(\xi_{jk}^A),s_{\min},s_{\max}\right).
\label{eq:effective_activation_scale}
\end{equation}
The implementations initialize $\xi_{jk}^{A,0}$ with either $\log s_{jk}^{*}$ or $\operatorname{softplus}^{-1}(s_{jk}^{*})$.
The former yields $\log(1+s_{jk}^{*})$ before clipping, while the latter recovers $s_{jk}^{*}$.
Section~\ref{subsec:experimental_setup} specifies the initializer used in each experimental setting.
We then enable activation quantization to obtain $f_{q^0}$, the initial model with quantized weights and activations in Fig.~\ref{fig:privacy_aware_ptq}.
In Sec.~\ref{subsec:quantizer_refinement}, we use a frozen copy to penalize inter-class expansion during utility recovery.

\subsection{Geometry-Regularized Quantizer Refinement}
\label{subsec:quantizer_refinement}

Starting from $f_{q^0}$, we refine quantization parameters to recover utility at the allocated precision.
As discussed in Sec.~\ref{subsec:data_dependent_inversion}, restoring class separation can also restore cues useful for identity recovery.
We therefore constrain inter-class expansion while jointly optimizing weight and activation scales and weight-rounding decisions, with bit widths and pretrained parameters fixed.

Two frozen references serve different roles: the FP32 model $f$ guides utility recovery, and $f_{q^0}$ anchors the geometry constraint.
On a class-balanced calibration mini-batch $\mathcal B$, we combine a generalized label-smoothing classification loss with prediction distillation:
\begin{equation}
\mathcal L_{\mathrm{CE}}
=-\mathbb E_{\mathcal B}\!\left[
(1-\eta)\log p_{q,y}
+\frac{\eta}{M}\sum_{c=1}^{M}\log p_{q,c}
\right],
\label{eq:refinement_classification_loss}
\end{equation}
\begin{equation}
\mathcal L_{\mathrm{KL}}
=\mathbb E_{\mathcal B}
D_{\mathrm{KL}}\!\left(\mathbf p_f(x)\,\|\,\mathbf p_q(x)\right),
\label{eq:refinement_distillation_loss}
\end{equation}
Here $p_{q,c}=p_{q,c}(x)$, $\mathbb E_{\mathcal B}$ averages over batch examples, and $M$ is the number of task classes.
For $\eta\in[0,1]$, the objective implements uniform label smoothing, with $\eta=0$ giving ordinary cross-entropy.
We also evaluate $\eta=-0.001$ as a signed label-sharpening variant whose coefficients do not define a probability target (Sec.~\ref{subsec:experimental_setup}).

To constrain the class structure restored during utility recovery, we connect the covariance analysis in Sec.~\ref{subsec:representation_geometry} and its class-wise decomposition in~\eqref{eq:class_covariance_decomposition} to a differentiable distance statistic.
For batch vectors $u_i$ from $K_{\mathcal B}\ge2$ classes, let $m_c$ be the class-$c$ mean and $\bar m$ the equally weighted mean of class centers.
Define the mean squared inter-class distance
\begin{equation}
D_{\mathrm{inter}}(\mathcal B;u)
=\frac{\sum_{c\ne c'}\|m_c-m_{c'}\|_2^2}
{K_{\mathcal B}(K_{\mathcal B}-1)}.
\label{eq:refinement_inter_class_distance}
\end{equation}
Let $S_{B,\mathcal B}^{u}=K_{\mathcal B}^{-1}\sum_c(m_c-\bar m)(m_c-\bar m)^\top$ denote the corresponding between-class covariance.
Since $\sum_c(m_c-\bar m)=0$, expanding the squared distances gives
\begin{equation}
D_{\mathrm{inter}}(\mathcal B;u)
=\frac{2K_{\mathcal B}}{K_{\mathcal B}-1}
\operatorname{tr}(S_{B,\mathcal B}^{u}).
\label{eq:refinement_distance_trace}
\end{equation}
Larger distances indicate greater spread of class centers and a larger between-class covariance trace.

We constrain class separation in output probabilities and a selected late-layer representation.
For the feature term, we flatten and $L_2$-normalize each sample's representation before computing class means.
For each space $r\in\{\mathrm{out},\mathrm{feat}\}$, let $d_r$ and $d_r^0$ be the distances from the current model and frozen $f_{q^0}$ on the same batch.
We penalize relative expansion using
\begin{equation}
\mathcal R_{\epsilon_{\mathrm{num}}}(d,d_0)
=\left[\max\!\left\{0,
\log\!\left(\frac{d}{d_0+\epsilon_{\mathrm{num}}}+\epsilon_{\mathrm{num}}\right)
\right\}\right]^2,
\label{eq:refinement_retention_penalty}
\end{equation}
\begin{equation}
\mathcal L_{\mathrm{geom}}
=\lambda_{\mathrm{out}}\mathcal R_{\epsilon_{\mathrm{num}}}(d_{\mathrm{out}},d_{\mathrm{out}}^0)
+\lambda_{\mathrm{feat}}\mathcal R_{\epsilon_{\mathrm{num}}}(d_{\mathrm{feat}},d_{\mathrm{feat}}^0).
\label{eq:refinement_geometry_loss}
\end{equation}
The quantized reference adapts the constraint to each model's starting geometry, and the distance ratio avoids a shared absolute-distance threshold across domains.
The stabilized one-sided penalty targets expansion beyond this reference without continually forcing smaller class distances to contract, leaving the task objectives room to recover utility.

We also penalize scale-parameter deviations from fixed references and encourage decisive rounding assignments.
These parameters determine the local quantization errors propagated in Sec.~\ref{subsec:prediction_information}.
Let $\xi_j^A$ collect the channel parameters in~\eqref{eq:effective_activation_scale}.
Let $s_\ell^{W,0}$ be the initial per-output-channel weight scales and set $s_\ell^W=s_\ell^{W,0}\exp(\delta_\ell^W)$ elementwise.
Weight-rounding logits $v_\ell^W$ define elementwise binary decisions $\mathbf1[v_\ell^W\ge0]$, with sigmoid surrogate $h_\ell=\sigma(v_\ell^W/\tau_r)$ at temperature $\tau_r>0$.
Floor codes computed at the initial weight scales are clipped to $[q_{\min},q_{\max}-1]$ and then fixed, preserving two integer candidates per weight during scale updates.
Adding the binary decisions, clipping to the assigned integer range, and multiplying by $s_\ell^W$ yields the quantized weights.
The forward pass uses these hard decisions, while the backward pass uses sigmoid derivatives; activation quantizers also use surrogate gradients.
Here, $\langle\cdot\rangle_{\mathcal A}$ and $\langle\cdot\rangle_{\mathcal W}$ average entries within each activation or weight quantizer, then across quantizers:
\begin{equation}
\begin{aligned}
\mathcal R_A&=\langle(\xi^A-\xi^{A,\mathrm{ref}})^2\rangle_{\mathcal A},\\
\mathcal R_W&=\langle(\delta^W)^2\rangle_{\mathcal W},\qquad
\mathcal R_{\mathrm{round}}=\langle h(1-h)\rangle_{\mathcal W},
\end{aligned}
\label{eq:refinement_parameter_penalties}
\end{equation}
where $\xi^{A,\mathrm{ref}}$ denotes the fixed reference parameters stored by the activation quantizers.
Their combined contribution is
\begin{equation}
\mathcal L_{\mathrm{reg}}
=\lambda_A\mathcal R_A+\lambda_W\mathcal R_W
+\lambda_r\mathcal R_{\mathrm{round}}.
\label{eq:refinement_regularization_loss}
\end{equation}

With $\theta_Q=\{\xi^A,\delta^W,v^W\}$, we jointly minimize $\lambda_{\mathrm{CE}}\mathcal L_{\mathrm{CE}}+\lambda_{\mathrm{KL}}\mathcal L_{\mathrm{KL}}+\mathcal L_{\mathrm{geom}}+\mathcal L_{\mathrm{reg}}$.
Scale parameters are projected onto fixed bounded intervals after each update, while the bit allocation from Sec.~\ref{subsec:bit_allocation_calibration} is retained.
In the cumulative ablations of Sec.~\ref{subsec:ablations}, every group optimizes scales and rounding decisions; the full combination of geometry and parameter penalties lowers strict attack success on Palm relative to CE--KL refinement.
Since the method adjusts the model's quantization parameters, it also supports composition with output-level post-processing defenses.
Section~\ref{subsec:experimental_setup} details optimization settings, feature locations, and calibration protocols.

\section{Experiments}
\label{sec:experiments}
We evaluate how the proposed PTQ method balances recognition utility and identity-recovery risk across data domains and model architectures.
We also assess its composition with output-level defenses and examine the effects of refinement objectives and calibration-data size.

% One-column setup summary at the top of the right column on page 12.
\begin{table}[!t]
\caption{Experimental setup for quantization and model inversion.}
\label{tab:experimental_setup}
\centering
\begingroup
\setlength{\tabcolsep}{2pt}
\renewcommand{\arraystretch}{1.00}
\setlength{\aboverulesep}{1.3pt}
\setlength{\belowrulesep}{1.3pt}
\fontsize{7.5pt}{8.6pt}\selectfont
\begin{tabular*}{\columnwidth}{@{\extracolsep{\fill}}ccccc@{}}
\toprule
\multicolumn{5}{l}{\textit{(a) Target-data split and auxiliary-pool size}} \\
Dataset & IDs & Train pool & Test & Auxiliary \\
\midrule
FaceScrub & 530 & 33,756 (80.00\%) & 8,440 (20.00\%) & 195,721 ($5.80\times$) \\
Palm & 50 & 617 (75.24\%) & 203 (24.76\%) & 4,682 ($7.59\times$) \\
Iris-Thousand & 150 & 2,250 (75.00\%) & 750 (25.00\%) & 17,000 ($7.56\times$) \\
\midrule
\end{tabular*}
\par\vspace{1pt}
\begin{tabular*}{\columnwidth}{@{\extracolsep{\fill}}ccccc@{}}
\multicolumn{5}{l}{\textit{(b) FP32 targets and independent evaluators}} \\
Dataset & Target & Acc. & Evaluator & Acc. \\
\midrule
FaceScrub & ResNet-50 & 78.98\% & ResNet-101 & 96.37\%$^{v}$ \\
Palm & ResNet-50 & 99.01\% & ResNet-101 & 96.06\% \\
Iris-Thousand & ResNet-50 & 95.20\% & ResNet-101 & 95.20\% \\
Iris-Thousand & Swin & 68.93\% & ResNet-101 & 95.20\% \\
\bottomrule
\end{tabular*}
\par\vspace{2pt}
\begin{minipage}{\columnwidth}
\fontsize{7.0pt}{8.0pt}\selectfont
\textit{Note:} Counts: images; IDs: target identities.
Train/test shares use $N_{\rm train}+N_{\rm test}$; auxiliary multiples use $N_{\rm train}$.
Face uses a separate CelebA auxiliary pool; Palm/Iris use 260/850 auxiliary identities disjoint from target IDs.
Acc. is image-level top-1 on the target tests in (a), except $^{v}$ (Face evaluator validation).
Evaluator fit/validation: 37,977/4,219 from all 42,196 Face images; 566/51 from Palm's training pool; 1,800/450 from Iris's 2,250 training images.
Both Iris targets use this split; their results in Figs.~\ref{fig:main_iris_architectures}--\ref{fig:ssd_composition} share the original ResNet-101 evaluator.
\end{minipage}
\endgroup
\end{table}

% Queued with Sec. VI-A so setup text can fill space below the plots.
\begin{figure*}[!t]
\centering
\includegraphics[width=\textwidth]{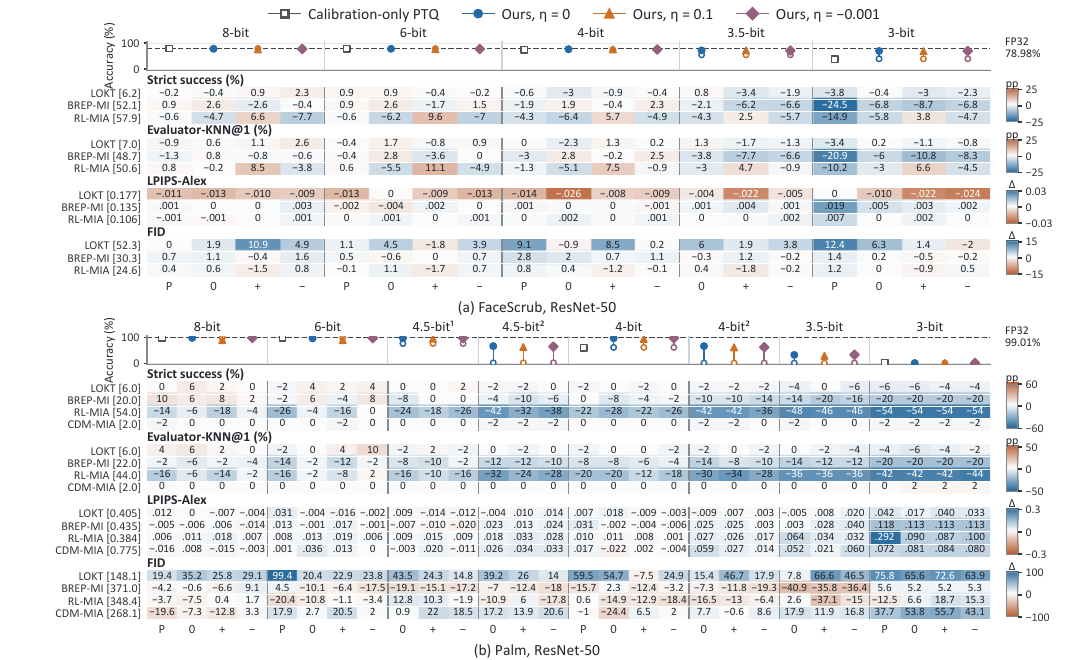}
\caption{\textbf{Utility and model inversion results on FaceScrub and Palm with ResNet-50 targets.}
Open circles and filled markers show accuracy before and after quantizer refinement; open squares (P) denote independent calibration-only PTQ references, and dashed lines indicate FP32 accuracy.
Heatmap entries are changes from the FP32 baseline for each attack and target; absolute baselines appear in brackets.
Changes in Strict success and Evaluator-KNN@1 are in percentage points; LPIPS-Alex and FID retain their original units.
Blue indicates lower identity-hit rates or larger distances, and orange indicates the reverse; each heatmap has its own zero-centered color scale.
$0$, $+$, and $-$ denote refinement with $\eta=0$, $0.1$, and $-0.001$, respectively.
Superscripts $1$ and $2$ distinguish mixed-precision candidate sets $\{6,5,4\}$ and $\{6,5,4,3\}$.
Palm additionally includes CDM-MIA.}
\label{fig:main_face_palm}
\end{figure*}

\begin{figure*}[!t]
\centering
\includegraphics[width=\textwidth]{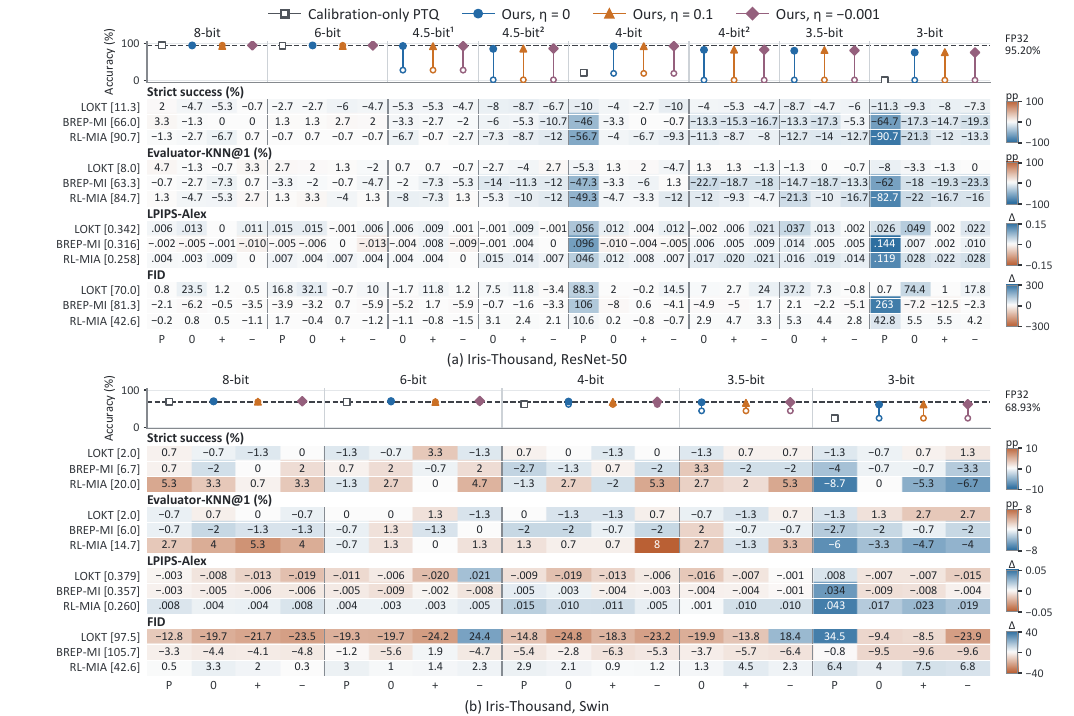}
\caption{\textbf{Utility and model inversion results on Iris-Thousand with ResNet-50 and Swin targets.}
Plotting conventions follow Fig.~\ref{fig:main_face_palm}.}
\label{fig:main_iris_architectures}
\end{figure*}

\subsection{Experimental Setup}
\label{subsec:experimental_setup}

\subsubsection{Datasets and Splits}
We study identity recovery on face, palmprint, and iris images; Table~\ref{tab:experimental_setup} summarizes the retained data, target/auxiliary identities, splits, and model accuracies.
\textbf{FaceScrub}~\cite{Ng_2014} contains over 100,000 face photographs of 530 individuals.
We resize the available images to $64\times64$ RGB without cropping.
The auxiliary pool comes from \textbf{CelebA}~\cite{Liu_2015_ICCV}, whose original release contains 202,599 images of 10,177 identities.
Following Yang et al.~\cite{Yang_2019}, we crop CelebA images to remove background before resizing them to the same resolution.
\textbf{Palm} uses CASIA-Palmprint, which contains 5,502 grayscale palmprint images.
\textbf{Iris-Thousand}, from CASIA-IrisV4, contains 20,000 near-infrared images from 1,000 subjects.
Both biometric inputs are resized to $224\times224$ with one channel.
The CASIA dataset links are given in Sec.~\ref{sec:introduction}.

\subsubsection{Attacks and Evaluation Protocol}
We retain one final candidate per identity, including failures, from three label-only attacks and one confidence-based black-box attack.
\textbf{LOKT}~\cite{NEURIPS2023_d9827e81} trains a target-specific conditional GAN for 20,000 iterations and a DenseNet-121 surrogate for 50 epochs, followed by five 600-step latent-optimization restarts.
\textbf{BREP-MI}~\cite{Kahla_2022_CVPR} searches a fixed domain-specific GAN prior with 100,000 optimization queries per identity, excluding initialization-pool queries.
Their limited recovery on Palm motivates testing diffusion-based \textbf{CDM-MIA}~\cite{10458692}.
It trains on hard-labeled auxiliary images for 100 epochs and selects among 12 candidates per identity using label consistency over 100 augmentations.
Its low recovery success (Fig.~\ref{fig:main_face_palm}) further supports the difficulty of label-only recovery on Palm discussed in Sec.~\ref{subsec:data_dependent_inversion}.
\textbf{RL-MIA}~\cite{Han_2023_CVPR} uses confidence-based rewards for up to 10,000 one-step episodes per identity.
Early stopping requires a successful candidate with probability at least $0.999999$, no improvement exceeding $10^{-7}$ for 1,000 episodes, and at least 80\% target hits over the latest 200 episodes.
Nearly all Palm and Iris runs reach the cap; Face stops early more often.
Appendix~J gives full protocols.
Complementing identity recovery, we evaluate \textbf{input reconstruction} from intermediate features using an attack design based on Liu et al.~\cite{liu2026partitioning} (Table~\ref{tab:input_reconstruction}).
% Detailed reconstruction analysis belongs in the planned Sec. VII discussion.

Following standard identity-recovery evaluation~\cite{Kahla_2022_CVPR,Han_2023_CVPR}, we train independent ResNet-101 evaluators from scratch on real labeled images, avoiding reliance on target predictions alone.
Their splits and accuracies appear in Table~\ref{tab:experimental_setup}; neither attack optimization nor candidate selection uses evaluator feedback.

\subsubsection{Target Models and Quantization Settings}
Using the targets in Table~\ref{tab:experimental_setup}, we compare FP32, independent calibration-only uniform W8A8/W6A6/W4A4/W3A3 PTQ, and our uniform or mixed-precision models (Figs.~\ref{fig:main_face_palm}--\ref{fig:main_iris_architectures}).
Floating-point quantize--dequantize operations use per-output-channel weight scales and per-channel activation scales, with separate weight and activation average-bit budgets (Sec.~\ref{subsec:bit_allocation_calibration}).
Fisher \mbox{estimation} uses up to 50 batches of 128/16/16/4 images for Face/Palm/Iris--ResNet-50/Iris--Swin; activation calibration uses the full calibration pool.
Activation parameters start at $\log s^*$ for ResNet and $\operatorname{softplus}^{-1}(s^*)$ for Swin; Appendix~\mbox  {J-A} gives fixed parameter references and scale bounds.
Feature geometry uses the activation after the second convolution of ResNet's final bottleneck, or Swin's pooled pre-head features.

Refinement runs for 100 epochs with Adam and learning rates $5\times10^{-5}$, $3\times10^{-5}$, and $10^{-5}$ starting at epochs 1, 41, and 76.
Class-balanced batches use $P$ classes and $K$ images per class, with $(P,K)=(32,4)$ for Face, $(10,2)$ for Palm/Iris--ResNet-50, and $(4,2)$ for Iris--Swin.
The CE/KL weights are 6/1, output/feature geometry weights are 2/0.5, and activation-scale/weight-scale/rounding regularization weights are $10^{-4}/10^{-5}/10^{-4}$.
We use $\eta\in\{0,0.1,-0.001\}$ and untempered teacher probabilities.
Refinement checkpoints are selected by hard-quantized test accuracy, so the reported utility is not an independent held-out estimate; calibration-only references have no such selection.
We also combine quantization with \textbf{SSD}~\cite{ICLR2025_1d8e261c}, using $T=0.1$ and validation-selected, frozen per-model $\epsilon$ values (Fig.~\ref{fig:ssd_composition}).

Our method requires representative calibration images with labels covering the target classes, but does not require them to belong to the original training set.
AdaRound~\cite{pmlr-v119-nagel20a} reports similar quantization accuracy using nontraining images from related domains.
We reuse the target-training pools to make full use of limited labeled data while preserving the attack auxiliary pools.
Table~\ref{tab:combined_ablations}(c) examines calibration-set size with class-covering subsets and a fixed refinement-update budget.

\subsubsection{Evaluation Metrics}
\textbf{Utility} is the target's top-1 recognition accuracy on real test images.
For recovered candidates, \textbf{Target} and \textbf{Evaluator} measure top-1 agreement with the intended identity; \textbf{Strict} requires both classifiers to agree with it.
\textbf{T-KNN@1} and \textbf{E-KNN@1} measure nearest-neighbor identity matches in the target and evaluator feature spaces, using squared Euclidean distance to real-image reference banks.
Learned perceptual image patch similarity (\textbf{LPIPS-Alex/VGG})~\cite{Zhang_2018_CVPR} averages each candidate's minimum \mbox{perceptual} distance to at most 20 real reference images of its intended identity.
Fr\'echet inception distance (\textbf{FID})~\cite{NIPS2017_8a1d6947} measures the distributional distance between recovered and real-image features.
Lower identity-hit rates and larger perceptual/distributional distances favor defense, interpreted jointly with utility.
Feature extraction, reference-bank construction, and distance computation are specified in Appendix~\mbox  {J-D}.
For input reconstruction, following prior evaluations~\cite{liu2026partitioning}, we use mean squared error (\textbf{MSE}), peak signal-to-noise ratio (\textbf{PSNR}), and structural similarity (\textbf{SSIM}) against the corresponding input; higher MSE and lower PSNR/SSIM favor defense.
Here, \textbf{LPIPS-Alex} also uses the corresponding input.

\subsection{Utility--Privacy Trade-offs under Quantization}
\label{subsec:main_results}

Figures~\ref{fig:main_face_palm} and~\ref{fig:main_iris_architectures} evaluate whether the design in Sec.~\ref{sec:privacy_aware_ptq} recovers recognition utility while limiting identity recovery.
The calibration-only references show why these objectives must be assessed together.
At 4 bits, ResNet-50 accuracy is $74.40\%$ on Face, $60.10\%$ on Palm, and $21.33\%$ on Iris, compared with FP32 accuracies of $78.98\%$, $99.01\%$, and $95.20\%$.
This divergence confirms the data-dependent sensitivity identified in Sec.~\ref{subsec:data_dependent_quantization} and supports examining the 4-bit regime highlighted by the analysis in Sec.~\ref{subsec:prediction_information}.
A large attack decrease is not useful when utility falls near chance, as for the 3-bit Palm and Iris references.

% Queue these figures with VI-B so its text can flow below them.
% Keep the two page-15 floats together instead of stretching their separation.
\begingroup
\makeatletter
\setlength{\@dblfpsep}{14pt}
\makeatother
\begin{figure*}[!t]
\centering
\includegraphics[width=\textwidth]{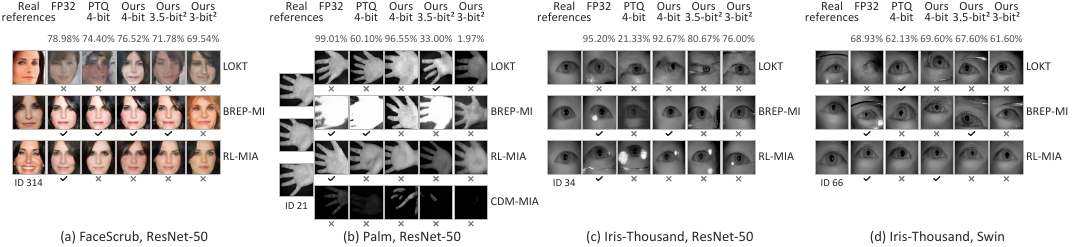}
\caption{\textbf{Qualitative identity recovery under quantization.}
Each panel shows three real images of one target identity and final attack candidates from FP32, calibration-only 4-bit PTQ, and our 4-, 3.5-, and 3-bit models ($\eta=0$).
The real images are identity references, not paired reconstruction targets.
Both 4-bit columns use uniform W4A4; candidate-set superscripts follow Fig.~\ref{fig:main_face_palm}.
Percentages are target-test accuracies.
Check marks denote Strict success (both the target and independent evaluator predict the target identity); crosses denote failure.
All shown identities satisfy FP32 Strict success for BREP-MI and RL-MIA.
The Iris examples additionally require Strict failure in at least two of the three displayed Ours configurations for each of these attacks; LOKT and CDM-MIA outcomes do not constrain selection.
The Iris panels use different identities.}
\label{fig:qualitative_identity_recovery}
\end{figure*}

\begin{figure*}[!t]
\centering
\includegraphics[width=\textwidth]{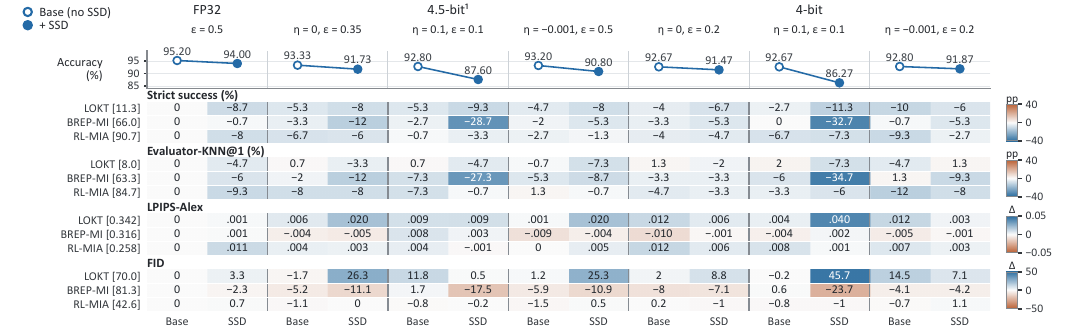}
\caption{\textbf{Composition with SSD~\cite{ICLR2025_1d8e261c} on Iris-Thousand with a ResNet-50 target.}
Open and filled circles show test accuracy without (Base) and with SSD for each matched checkpoint.
SSD uses $T=0.1$ and the displayed per-model $\epsilon$, frozen under model-specific validation budgets.
Accuracy uses one response per test image (750 images).
Heatmaps use each attack's original FP32 baseline without SSD.
Superscript definitions and other heatmap conventions follow Fig.~\ref{fig:main_face_palm}.
The 4-bit models use uniform W4A4.
Results include all 150 final candidates per non-adaptive attack, retaining failures.}
\label{fig:ssd_composition}
\end{figure*}
\endgroup

Figure~\ref{fig:main_face_palm}(b) demonstrates a useful trade-off at uniform 4 bits.
With $\eta=0$, our Palm model improves accuracy from $61.58\%$ before refinement to $96.55\%$, while BREP-MI and RL-MIA Strict success is $12\%$ and $26\%$, below their FP32 values of $20\%$ and $54\%$.
It also improves on the independent 4-bit reference's $60.10\%$ accuracy and $20\%/32\%$ Strict rates.
Thus, utility recovery need not restore the original identity-recovery risk, directly validating the intended trade-off of Sec.~\ref{subsec:quantizer_refinement} in this setting.
The RL-MIA E-KNN rate also falls from $44\%$ to $24\%$, corroborating the reduction beyond the evaluator's classification head.
In Fig.~\ref{fig:main_iris_architectures}(a), uniform 4-bit refinement restores Iris accuracy from $19.60\%$ to $92.67\%$, close to FP32's $95.20\%$.
Its BREP-MI/RL-MIA Strict rates remain $62.67\%/86.67\%$, however, showing that strong utility recovery alone does not ensure a large privacy gain.

The mixed-precision results reveal additional operating points and their domain dependence.
In Fig.~\ref{fig:main_iris_architectures}(a), Iris at 3.5 bits with $\eta=0.1$ reaches $82.53\%$ accuracy while reducing BREP-MI/RL-MIA Strict success to $48.67\%/76.67\%$, from FP32's $66\%/90.67\%$.
RL-MIA also shows lower E-KNN and higher LPIPS/FID, supporting reduced recovery across complementary metrics.
At 3 bits with $\eta=0$, Iris retains $76\%$ accuracy and lowers RL-MIA Strict success to $69.33\%$.
Palm responds differently: its 3.5-bit accuracy remains at most $33.50\%$, and 3-bit refinement stays at $1.97\%$.
These results support the role of data-dependent quantization response in Secs.~\ref{subsec:representation_geometry} and~\ref{subsec:data_dependent_quantization}: the usable precision range must be established for each domain.

Figure~\ref{fig:main_face_palm}(a) further shows that the refinement objective changes the trade-off at a fixed bit width.
For Face at 3.5 bits, $\eta=-0.001$ yields $71.07\%$ accuracy and BREP-MI/RL-MIA Strict rates of $45.47\%/52.26\%$, below FP32's $52.08\%/57.92\%$.
Using $\eta=0.1$ gives similar accuracy ($71.73\%$), but RL-MIA Strict success rises to $60.38\%$, although LOKT and BREP-MI still weaken relative to FP32.
RL-MIA can still exploit relative confidence values even if smoothing reduces peak probabilities, so weaker label-only recovery need not imply less effective confidence-based search.
Similar task accuracy at the same bit width can therefore accompany substantially different inversion risk, \mbox{consistent} with the configuration-dependent analysis in Sec.~\ref{subsec:prediction_information}.
The subsequent cumulative ablations examine the contribution of the geometry and parameter penalties.

Figure~\ref{fig:main_iris_architectures}(b) tests transfer to a Swin target.
At 3 bits with $\eta=-0.001$, our model reaches $62.27\%$ accuracy, compared with $24.93\%$ for calibration-only PTQ, while BREP-MI/RL-MIA Strict success is $3.33\%/13.33\%$, below FP32's $6.67\%/20\%$.
This demonstrates that the method's utility-recovery capability extends to the Transformer target, with architecture- and attack-dependent privacy gains.
Re-evaluating the same candidates with Swin-B changes absolute success rates and some FP32-relative trends; Appendix~\mbox  {K-D} examines this sensitivity, and Appendix~K gives all numerical results.

Finally, Fig.~\ref{fig:qualitative_identity_recovery} illustrates these changes for selected identities.
Candidates can preserve recognizable face or texture structure while failing target-identity recognition, and success varies across attacks at the same precision.
These examples complement the aggregate Strict and E-KNN measurements.
Taken together, Figs.~\ref{fig:main_face_palm}--\ref{fig:qualitative_identity_recovery} support the data- and model-dependent relationship established in Secs.~\ref{sec:information_theory}--\ref{sec:data_effects} and demonstrate that the proposed refinement can recover substantial utility while retaining reduced identity-recovery risk.
The remaining risk motivates combining PTQ with SSD.

\subsection{Composition with Output-Level Defense}
\label{subsec:ssd_composition}

Figure~\ref{fig:ssd_composition} evaluates SSD~\cite{ICLR2025_1d8e261c} as a complement to the quantizer refinement in Sec.~\ref{sec:privacy_aware_ptq}.
SSD post-processes predictions without further model training.
We rerun all three attacks on seven matched checkpoints with and without SSD, following Appendix~J.
On FP32, protection against BREP-MI and RL-MIA remains limited: SSD changes their Strict success from $66\%/90.67\%$ to $65.33\%/82.67\%$, with accuracy decreasing from $95.2\%$ to $94\%$.

% Queue the ablation and reconstruction tables for page 16.
% Remaining ablation and reconstruction tables reserved for page 16.
\begin{table*}[!t]
\caption{Objective and calibration-data ablations on ResNet-50 ($\eta=0$).
Arrows indicate defender-preferred values.}
\label{tab:combined_ablations}
\centering
\begingroup
\setlength{\tabcolsep}{2.0pt}
\renewcommand{\arraystretch}{1.0}
\setlength{\aboverulesep}{0.5pt}
\setlength{\belowrulesep}{0.5pt}
\fontsize{6.0pt}{6.2pt}\selectfont
\begin{tabular*}{\textwidth}{@{\extracolsep{\fill}}ccccc c ccccc ccc@{}}
\toprule
\multirow{2}{*}{Variant} & \multicolumn{2}{c}{Accuracy (\%)} & \multicolumn{2}{c}{Change (pp)} & \multirow{2}{*}{Attack} & \multicolumn{5}{c}{Identity rates (\%)} & \multicolumn{2}{c}{LPIPS} & \multirow{2}{*}{FID $\uparrow$} \\
\cmidrule(lr){2-5}\cmidrule(lr){7-13}
& Pre $\uparrow$ & Post $\uparrow$ & $\Delta$ & $\Delta_{\mathrm{FP32}}$ & & Target $\downarrow$ & Eval. $\downarrow$ & Strict $\downarrow$ & T-KNN $\downarrow$ & E-KNN $\downarrow$ & Alex $\uparrow$ & VGG $\uparrow$ & \\
\midrule
\multicolumn{14}{l}{\textit{(a) Palm: cumulative objective ablation, uniform W4A4}} \\
\multirow{3}{*}{A} & \multirow{3}{*}{$61.58$} & \multirow{3}{*}{$96.06$} & \multirow{3}{*}{$34.48$} & \multirow{3}{*}{$-2.96$} & LOKT~\cite{NEURIPS2023_d9827e81} & $30.00$ & $6.00$ & $4.00$ & $24.00$ & $4.00$ & $0.399$ & $0.522$ & $168.11$ \\
 &  &  &  &  & BREP-MI~\cite{Kahla_2022_CVPR} & $94.00$ & $20.00$ & $20.00$ & $68.00$ & $20.00$ & $0.418$ & $0.558$ & $372.60$ \\
 &  &  &  &  & RL-MIA~\cite{Han_2023_CVPR} & $64.00$ & $28.00$ & $28.00$ & $64.00$ & $24.00$ & $0.389$ & $0.541$ & $328.01$ \\
\cmidrule(lr){6-14}
\multirow{3}{*}{AB} & \multirow{3}{*}{$61.58$} & \multirow{3}{*}{$96.55$} & \multirow{3}{*}{$34.98$} & \multirow{3}{*}{$-2.46$} & LOKT & $46.00$ & $12.00$ & $10.00$ & $24.00$ & $16.00$ & $0.403$ & $0.519$ & $166.08$ \\
 &  &  &  &  & BREP-MI & $92.00$ & $24.00$ & $24.00$ & $60.00$ & $18.00$ & $0.424$ & $0.562$ & $364.17$ \\
 &  &  &  &  & RL-MIA & $62.00$ & $28.00$ & $28.00$ & $62.00$ & $32.00$ & $0.389$ & $0.539$ & $337.61$ \\
\cmidrule(lr){6-14}
\multirow{3}{*}{ABC} & \multirow{3}{*}{$61.58$} & \multirow{3}{*}{$96.55$} & \multirow{3}{*}{$34.98$} & \multirow{3}{*}{$-2.46$} & LOKT & $30.00$ & $8.00$ & $6.00$ & $22.00$ & $4.00$ & $0.399$ & $0.524$ & $187.91$ \\
 &  &  &  &  & BREP-MI & $90.00$ & $20.00$ & $20.00$ & $60.00$ & $18.00$ & $0.435$ & $0.566$ & $359.80$ \\
 &  &  &  &  & RL-MIA & $66.00$ & $34.00$ & $34.00$ & $64.00$ & $28.00$ & $0.380$ & $0.533$ & $345.51$ \\
\cmidrule(lr){6-14}
\multirow{3}{*}{ABCD} & \multirow{3}{*}{$61.58$} & \multirow{3}{*}{$96.55$} & \multirow{3}{*}{$34.98$} & \multirow{3}{*}{$-2.46$} & LOKT & $30.00$ & $8.00$ & $8.00$ & $24.00$ & $12.00$ & $0.400$ & $0.520$ & $165.91$ \\
 &  &  &  &  & BREP-MI & $92.00$ & $12.00$ & $10.00$ & $64.00$ & $14.00$ & $0.432$ & $0.565$ & $365.72$ \\
 &  &  &  &  & RL-MIA & $62.00$ & $26.00$ & $24.00$ & $62.00$ & $30.00$ & $0.381$ & $0.537$ & $337.36$ \\
\midrule
\multicolumn{14}{l}{\textit{(b) Iris-Thousand: cumulative objective ablation, 4-bit mixed precision}} \\
\multirow{3}{*}{A} & \multirow{3}{*}{$1.47$} & \multirow{3}{*}{$88.80$} & \multirow{3}{*}{$87.33$} & \multirow{3}{*}{$-6.40$} & LOKT~\cite{NEURIPS2023_d9827e81} & $26.00$ & $10.00$ & $4.67$ & $20.00$ & $9.33$ & $0.386$ & $0.484$ & $142.12$ \\
 &  &  &  &  & BREP-MI~\cite{Kahla_2022_CVPR} & $97.33$ & $61.33$ & $61.33$ & $82.00$ & $52.00$ & $0.320$ & $0.441$ & $80.07$ \\
 &  &  &  &  & RL-MIA~\cite{Han_2023_CVPR} & $100.00$ & $85.33$ & $85.33$ & $99.33$ & $79.33$ & $0.263$ & $0.394$ & $42.20$ \\
\cmidrule(lr){6-14}
\multirow{3}{*}{AB} & \multirow{3}{*}{$1.47$} & \multirow{3}{*}{$89.47$} & \multirow{3}{*}{$88.00$} & \multirow{3}{*}{$-5.73$} & LOKT & $26.67$ & $5.33$ & $2.67$ & $16.00$ & $9.33$ & $0.352$ & $0.452$ & $65.94$ \\
 &  &  &  &  & BREP-MI & $98.67$ & $63.33$ & $63.33$ & $83.33$ & $55.33$ & $0.314$ & $0.438$ & $81.15$ \\
 &  &  &  &  & RL-MIA & $100.00$ & $88.00$ & $88.00$ & $98.67$ & $80.67$ & $0.260$ & $0.392$ & $41.98$ \\
\cmidrule(lr){6-14}
\multirow{3}{*}{ABC} & \multirow{3}{*}{$1.47$} & \multirow{3}{*}{$83.20$} & \multirow{3}{*}{$81.73$} & \multirow{3}{*}{$-12.00$} & LOKT & $22.67$ & $10.67$ & $4.67$ & $20.67$ & $10.67$ & $0.350$ & $0.453$ & $76.98$ \\
 &  &  &  &  & BREP-MI & $96.00$ & $60.00$ & $59.33$ & $65.33$ & $55.33$ & $0.323$ & $0.444$ & $81.77$ \\
 &  &  &  &  & RL-MIA & $100.00$ & $78.67$ & $78.67$ & $92.67$ & $68.67$ & $0.277$ & $0.403$ & $45.20$ \\
\cmidrule(lr){6-14}
\multirow{3}{*}{ABCD} & \multirow{3}{*}{$1.47$} & \multirow{3}{*}{$83.20$} & \multirow{3}{*}{$81.73$} & \multirow{3}{*}{$-12.00$} & LOKT & $26.67$ & $6.67$ & $4.00$ & $16.67$ & $4.00$ & $0.359$ & $0.458$ & $92.56$ \\
 &  &  &  &  & BREP-MI & $97.33$ & $51.33$ & $51.33$ & $67.33$ & $46.00$ & $0.324$ & $0.445$ & $84.40$ \\
 &  &  &  &  & RL-MIA & $100.00$ & $80.00$ & $80.00$ & $93.33$ & $68.00$ & $0.272$ & $0.402$ & $44.84$ \\
\midrule
\multicolumn{14}{l}{\textit{(c) Iris-Thousand: calibration-data amount, 4-bit mixed precision}} \\
\multirow{3}{*}{25\% ($n=562$)} & \multirow{3}{*}{$1.20$} & \multirow{3}{*}{$81.47$} & \multirow{3}{*}{$80.27$} & \multirow{3}{*}{$-13.73$} & LOKT~\cite{NEURIPS2023_d9827e81} & $26.00$ & $12.00$ & $5.33$ & $12.67$ & $8.67$ & $0.345$ & $0.453$ & $71.82$ \\
 &  &  &  &  & BREP-MI~\cite{Kahla_2022_CVPR} & $91.33$ & $52.67$ & $52.67$ & $60.67$ & $47.33$ & $0.322$ & $0.442$ & $76.74$ \\
 &  &  &  &  & RL-MIA~\cite{Han_2023_CVPR} & $100.00$ & $76.67$ & $76.67$ & $92.00$ & $71.33$ & $0.273$ & $0.402$ & $46.53$ \\
\cmidrule(lr){6-14}
\multirow{3}{*}{50\% ($n=1125$)} & \multirow{3}{*}{$1.07$} & \multirow{3}{*}{$80.67$} & \multirow{3}{*}{$79.60$} & \multirow{3}{*}{$-14.53$} & LOKT & $20.67$ & $6.00$ & $4.00$ & $10.00$ & $6.00$ & $0.346$ & $0.457$ & $81.08$ \\
 &  &  &  &  & BREP-MI & $96.00$ & $55.33$ & $55.33$ & $66.67$ & $46.00$ & $0.320$ & $0.443$ & $83.81$ \\
 &  &  &  &  & RL-MIA & $100.00$ & $82.00$ & $82.00$ & $93.33$ & $72.00$ & $0.276$ & $0.404$ & $48.98$ \\
\cmidrule(lr){6-14}
\multirow{3}{*}{75\% ($n=1688$)} & \multirow{3}{*}{$1.33$} & \multirow{3}{*}{$83.60$} & \multirow{3}{*}{$82.27$} & \multirow{3}{*}{$-11.60$} & LOKT & $19.33$ & $8.67$ & $4.00$ & $11.33$ & $6.67$ & $0.354$ & $0.459$ & $73.39$ \\
 &  &  &  &  & BREP-MI & $94.67$ & $58.67$ & $58.00$ & $63.33$ & $51.33$ & $0.316$ & $0.439$ & $76.13$ \\
 &  &  &  &  & RL-MIA & $100.00$ & $71.33$ & $71.33$ & $96.00$ & $68.67$ & $0.276$ & $0.403$ & $47.26$ \\
\cmidrule(lr){6-14}
\multirow{3}{*}{100\%$^{\dagger}$ ($n=2250$)} & \multirow{3}{*}{$1.47$} & \multirow{3}{*}{$82.93$} & \multirow{3}{*}{$81.47$} & \multirow{3}{*}{$-12.27$} & LOKT & $32.00$ & $11.33$ & $7.33$ & $15.33$ & $9.33$ & $0.341$ & $0.450$ & $77.00$ \\
 &  &  &  &  & BREP-MI & $94.67$ & $52.67$ & $52.67$ & $63.33$ & $40.67$ & $0.323$ & $0.441$ & $76.48$ \\
 &  &  &  &  & RL-MIA & $100.00$ & $79.33$ & $79.33$ & $93.33$ & $72.67$ & $0.275$ & $0.402$ & $45.49$ \\
\bottomrule
\end{tabular*}
\par\vspace{3pt}
\begin{minipage}{\textwidth}
\fontsize{6.0pt}{6.2pt}\selectfont
\textit{Note:} Pre/Post: accuracy before/after refinement; $\Delta=\mathrm{Post}-\mathrm{Pre}$ and $\Delta_{\mathrm{FP32}}=\mathrm{Post}-\mathrm{FP32}$ (pp).
A: classification; B: distillation; C: geometry retention; D: scale/rounding regularization; all groups optimize scales and rounding.
Target/Eval.: corresponding top-1 rates; Strict: simultaneous hits; T/E-KNN: KNN@1 in the corresponding feature space.
LPIPS uses the nearest target-identity real image; FID is a set-level distance.
All 50 Palm or 150 Iris candidates are included.
Iris uses 4-bit weight/activation budgets with candidates $\{6,5,4,3\}$.
In (c), the 25--75\% subsets cover all 150 classes and are shared by all three stages, with a fixed refinement-update budget.
$^{\dagger}$Historical reference with different duplicate/overlap filtering, not a matched 100\% control.
Test accuracy was used for checkpoint selection.
\end{minipage}
\endgroup
\end{table*}

% Independent single-column float: leave the right column available for text.
\begin{table}[!t]
% Independent single-column table placed below Table II on page 16.
% Source: 报告/reconstruction_attack_record_corrected_loader_ls_comparison_with_fp32_zh.md
% FP32 rows use the complete four-metric rows in Sec. 0.1, not the older Sec. 0.2 row.
% Original evaluation outputs have not been recovered; source details remain in
% reconstruction_results_table.source.json. Candidate-set superscripts are omitted
% because this summary does not identify its checkpoints by path or hash.
\caption{Input reconstruction on FaceScrub (ResNet-50).}
\label{tab:input_reconstruction}
\centering
\begingroup
\setlength{\tabcolsep}{2pt}
\renewcommand{\arraystretch}{1.0}
\setlength{\aboverulesep}{0.5pt}
\setlength{\belowrulesep}{0.5pt}
\fontsize{6.0pt}{6.2pt}\selectfont
\begin{tabular*}{\linewidth}{@{\extracolsep{\fill}}cccccc@{}}
\toprule
Precision & $\eta$ & MSE $\uparrow$ & PSNR (dB) $\downarrow$ & SSIM $\downarrow$ & LPIPS-Alex $\uparrow$ \\
\midrule
\multicolumn{6}{l}{\textit{(a) feat20}} \\
FP32 & --- & $0.008440$ & $25.524$ & $0.9568$ & $0.0965$ \\
3.5-bit & $0.1$ & $0.011302$ & $24.250$ & $0.9409$ & $0.1226$ \\
3.5-bit & $-0.001$ & $0.011493$ & $24.177$ & $0.9393$ & $0.1217$ \\
3-bit & $0.1$ & $0.011563$ & $24.151$ & $0.9398$ & $0.1199$ \\
3-bit & $-0.001$ & $0.011546$ & $24.158$ & $0.9392$ & $0.1220$ \\
\midrule
\multicolumn{6}{l}{\textit{(b) feat46}} \\
FP32 & --- & $0.075299$ & $16.023$ & $0.6215$ & $0.2519$ \\
3.5-bit & $0.1$ & $0.085650$ & $15.466$ & $0.5620$ & $0.3008$ \\
3.5-bit & $-0.001$ & $0.075929$ & $15.984$ & $0.6032$ & $0.2783$ \\
3-bit & $0.1$ & $0.082937$ & $15.601$ & $0.5628$ & $0.3235$ \\
3-bit & $-0.001$ & $0.076033$ & $15.979$ & $0.6104$ & $0.2813$ \\
\bottomrule
\end{tabular*}
\par\vspace{3pt}
\begin{minipage}{\linewidth}
\fontsize{6.0pt}{6.2pt}\selectfont
\textit{Note:} Metrics compare reconstructions with their corresponding inputs.
$\eta$ denotes label smoothing during refinement; arrows indicate defender-preferred values.
\end{minipage}
\endgroup

\end{table}

Much larger BREP-MI reductions occur for quantized models refined with label smoothing ($\eta=0.1$).
At 4.5 bits with $\epsilon=0.1$, adding SSD lowers BREP-MI Strict success from $63.33\%$ to $37.33\%$, while accuracy decreases from $92.8\%$ to $87.6\%$.
E-KNN also falls from $56\%$ to $36\%$, supporting weaker identity recovery beyond changes in the target's final decisions.
At uniform 4 bits with the same $\eta$ and $\epsilon$, BREP-MI Strict success decreases from $66\%$ to $33.33\%$ and LOKT from $8.67\%$ to $0\%$, with accuracy changing from $92.67\%$ to $86.27\%$.
Thus, SSD can provide substantial additional identity protection after privacy-aware PTQ with label smoothing.

The interaction between label smoothing and SSD's perturbation rule helps explain this pattern.
Positive smoothing discourages sharply peaked predictions; the saved test outputs of the $\eta=0.1$ models indeed have higher entropy and smaller top-two probability gaps than FP32 on average.
SSD's entropy-scaled $\ell_1$ budget, $\epsilon H(\mathbf p)$, then permits larger perturbations, while smaller gaps require less probability redistribution to change the returned label.
Smaller gaps also increase the relative probability that SSD conditions on a competing class.
Existing validation sweeps support this mechanism: at the same $\epsilon=0.1$, expected label-flip rates are $0.26\%$ for FP32 versus $5.03\%$ and $5.04\%$ for the 4.5- and 4-bit $\eta=0.1$ models.
This greater susceptibility to label changes is also consistent with their larger utility cost; Appendix~\mbox  {L-C} gives the output statistics and sweep results.

BREP-MI Strict success decreases in all seven paired comparisons, whereas RL-MIA responds less consistently.
For example, at uniform 4 bits with $\eta=-0.001$, adding SSD increases RL-MIA Strict success from $81.33\%$ to $88\%$.
The added protection is therefore configuration- and attack-dependent.
The attack comparisons also use model-specific utility budgets, rather than a common defense setting; Appendix~L gives the complete paired results.

\subsection{Ablation Studies}
\label{subsec:ablations}

Table~\ref{tab:combined_ablations} examines the privacy constraints in Sec.~\ref{subsec:quantizer_refinement} and calibration-data quantity.
Within each dataset, cumulative objective variants share initialization and bit allocation, optimizing scales and rounding throughout.

On Palm (Table~\ref{tab:combined_ablations}(a)), AB, ABC, and ABCD all attain $96.55\%$ accuracy.
Relative to CE--KL refinement (AB), the complete objective (ABCD) reduces LOKT/BREP-MI/RL-MIA Strict success from $10\%/24\%/28\%$ to $8\%/10\%/24\%$, with lower E-KNN across all three attacks.
This supports the intended privacy improvement at matched utility.
Individual additions remain attack-dependent: geometry retention raises RL-MIA Strict success from $28\%$ to $34\%$, while the subsequent parameter penalties raise LOKT from $6\%$ to $8\%$.
The benefit therefore comes from the combined constraints, without a uniform gain from every term.

On Iris (Table~\ref{tab:combined_ablations}(b)), AB to ABCD lowers BREP-MI/RL-MIA Strict success from $63.33\%/88\%$ to $51.33\%/80\%$, but accuracy falls from $89.47\%$ to $83.20\%$.
At unchanged $83.20\%$ accuracy, ABC to ABCD lowers BREP-MI Strict success from $59.33\%$ to $51.33\%$, although RL-MIA rises slightly.
Parameter regularization thus improves BREP-MI resistance at matched accuracy, with attack-dependent effects.

Table~\ref{tab:combined_ablations}(c) uses class-covering subsets and a fixed refinement-update budget.
Only $25\%$ of the original pool ($562$ images) already yields $81.47\%$ accuracy; $50\%$ and $75\%$ yield $80.67\%$ and $83.60\%$.
Additional data therefore brings limited, nonmonotonic utility gains over this range.
BREP-MI Strict success increases with pool size, while RL-MIA is nonmonotonic, so larger pools do not consistently improve attack resistance either.
The historical $100\%$ row is not a matched control because its duplicate and overlap filtering differs.
Appendix~M gives the detailed controls and additional metric comparisons.

\section{Discussion}
\label{sec:discussion}
We discuss input reconstruction and the storage and computational costs of post-training quantization.

\subsection{Input Reconstruction}
\label{subsec:input_reconstruction}
Table~\ref{tab:input_reconstruction} complements identity recovery with FaceScrub input reconstruction from intermediate features.
Under FP32, SSIM is $0.9568$ at feat20 but $0.6215$ at feat46, demonstrating the baseline difficulty difference discussed in Sec.~\ref{subsec:data_dependent_inversion}.

All tested quantized configurations worsen all four reconstruction metrics relative to their same-layer FP32 baseline.
At feat46, $\eta=0.1$ impairs reconstruction more than $\eta=-0.001$ at both precisions, whereas feat20 differences are small and metric-dependent.
This layer dependence is consistent with the Neural Vortex interpretation of Liu et al.~\cite{liu2026partitioning}: label smoothing can increase output entropy while tightening within-class geometry in decision-oriented layers, potentially hindering input reconstruction.
Moving from 3.5 to 3 bits with $\eta=0.1$ improves feat46's pixel-based metrics but worsens LPIPS.

The Palm examples in Appendix~N retain blurred palm-line structure at feat3 even with 3-bit activation quantization.
This shallow cut leaves most backbone layers on the cloud and produces a large feature map; our collaborative-inference analysis therefore focuses on feat21 and deeper cuts.
Only faint palm-line traces remain at feat21, while later examples, including FP32, mainly preserve hand contours.
FP32 prediction-based reconstructions likewise lack clear palm-line detail.
These observations qualitatively complement Fig.~\ref{fig:data_characteristics_quantization}(b).

\subsection{Storage and Computational Cost}
\label{subsec:deployment_cost}
We use floating-point quantize--dequantize operations~\cite{pmlr-v119-nagel20a,2021-iclr-brecq}, as native support for flexible layer-wise precisions and scale granularities remains backend-dependent~\cite{pmlr-v139-yao21a,Wei_2025}.
The complete offline quantization workflow is directly executable and timed, whereas simulated inference incurs floating-point computation and quantize--dequantize overhead.
It therefore does not yield representative native low-bit inference latency, which requires compatible kernels and hardware-specific benchmarking.

Ideal packed-weight storage is $\sum_\ell n_\ell b_\ell^W/8$ bytes, or $\bar b_W/32$ of the corresponding FP32 weight payload (Sec.~\ref{subsec:bit_allocation_calibration}), excluding unquantized parameters, scales, and format overhead.
Uniform 4-bit quantization gives an eightfold payload reduction, with estimated sizes of $11.70$, $11.23$, $11.33$, and $41.34$~MiB for Face--ResNet-50, Palm--ResNet-50, Iris--ResNet-50, and Iris--Swin, respectively.

For W4A4 with $\eta=0$ and 100 refinement epochs, complete offline workflows took 6~min~19~s for Palm--ResNet-50 and 20~min~17~s for Iris--ResNet-50, including quantization, evaluation, analysis, and final checks.

\section{Conclusion}
\label{sec:conclusion}
We investigated the relationship between model quantization and MIAs through information-theoretic analysis and cross-domain experiments.
Our analysis bounds changes in categorical prediction information and characterizes label stability and representation perturbations.
The experiments distinguish baseline inversion difficulty from quantization-induced changes and reveal pronounced domain sensitivity differences at 4 bits.
Guided by these findings, our privacy-aware mixed-precision PTQ combines sensitivity-guided bit allocation with geometry-regularized scale and rounding refinement to recover utility while limiting identity recovery.
Evaluations across face, palmprint, and iris recognition demonstrate useful utility--privacy trade-offs and additional protection through composition with output-level defenses.
These findings support assessing quantization jointly with data characteristics, the attack interface, and retained recognition utility.

\bibliographystyle{IEEEtranTPAMI}
\bibliography{references_strings,references}

\end{document}